\documentclass[a4paper]{lipics-v2021}
\usepackage{xspace,amsmath,amsfonts,bm,bbm,ifthen,mathtools}
\usepackage{amssymb}
\usepackage{mathtools}
\usepackage[heavycircles]{stmaryrd}
\usepackage[sort, noadjust]{cite}
\usepackage{mathrsfs}
\usepackage{leftidx}
\usepackage[scaled=.89]{helvet}
\usepackage{tikz}
\usepackage{multicol}
\usepackage{amsmath}
\hideLIPIcs
\nolinenumbers

\AddToHook{cmd/appendix/before}{
	\setcounter{equation}{0}
	
	\setcounter{theorem}{0}

}

\usetikzlibrary{arrows, automata, shapes, calc, fadings}
\tikzset{->, auto, >=stealth', font=\small}
\tikzset{state/.style={shape=circle, draw, fill=white, initial text=, inner sep=.5mm, minimum size=1.5mm}}
\tikzset{accepting/.style=accepting by arrow}
\tikzset{state with output/.style={shape=rectangle split, rectangle split parts=2, draw, fill=white, initial text=, inner sep=1mm}}

\newcommand{\ms}[1]{\mathsf{#1}}
\newcommand{\mc}[1]{\mathcal{#1}}

\newcommand{\sth}{\text{ s.t.\ }}
\newcommand{\+}{\text{+}}

\newcommand*\rest[1]{_{\upharpoonleft #1}}
\newcommand{\Nat}{\ensuremath{\mathbb{N}}}
\newcommand{\Int}{\ensuremath{\mathbb{Z}}}

\newcommand{\maps}{\rightarrow}
\newcommand{\equivalent}{\Leftrightarrow}
\newcommand{\midbar}{\mathrel{\big|}}
\newcommand{\set}[1]{{\{#1\}}}

\renewcommand{\subset}{\subseteq}
\renewcommand{\geq}{\geqslant}
\renewcommand{\leq}{\leqslant}
\renewcommand{\implies}{\Rightarrow}
\renewcommand{\inf}[1]{\ms{inf}\set{#1}}
\renewcommand{\sup}[1]{\ms{sup}\set{#1}}

\newcommand{\conc}{\square}
\newcommand{\Id}{\ms{Id}}

\newcommand{\ipoms}{\ms{iPoms}}
\newcommand{\wipoms}{{\ms{iPoms}^{\omega}}}
\newcommand{\infipoms}{{\ms{iPoms}^{\infty}}}
\newcommand{\iipoms}{\ms{iiPoms}}
\newcommand{\wiipoms}{{\ms{iiPoms}^{\omega}}}
\newcommand{\infiipoms}{{\ms{iiPoms}^{\infty}}}

\newcommand{\coh}{\ms{Coh}^+}
\newcommand{\wcoh}{\ms{Coh}^\omega}
\newcommand{\infcoh}{\ms{Coh}^\infty}

\newcommand{\rat}{\mc{R}at}
\newcommand{\Drat}{D\text{-}\mc{R}at}
\newcommand{\wrat}{{\mc{R}at}^\omega}
\newcommand{\infrat}{{\mc{R}at}^\infty}

\newcommand{\Xib}{\overline{\Xi}}

\newcommand{\id}{\ms{id}}
\newcommand{\ilo}[3]{\leftidx{_{#1}}{#2}{_{#3}}}

\newcommand{\ibu}{\vcenter{\hbox{\tiny $\bullet$}}}
\newcommand{\nibu}{\hphantom{\ibu}}

\newcommand{\starter}[2]{{}_{#2\!}{\uparrow}{#1}}
\newcommand{\terminator}[2]{{#1}{\downarrow}_{#2}}

\newcommand{\loset}[1]{\left[\begin{smallmatrix}#1\end{smallmatrix}\right]}

\newcommand{\subsu}{\sqsubseteq}

\newcommand{\down}{\mathord{\downarrow}}
\newcommand{\sd}{_{\down}}

\newcommand{\eps}{\id_\emptyset}
\renewcommand{\phi}{\varphi}

\newcommand{\evord}{\dashrightarrow}

\newcommand{\src}{\ms{src}}
\newcommand{\tgt}{\ms{tgt}}
\newcommand{\ev}{\ms{ev}}
\newcommand{\arrO}[1]{\mathrel{\nearrow^{#1}}}
\newcommand{\arrI}[1]{\mathrel{\searrow_{#1}}}

\newcommand{\prf}{\prec}
\newcommand{\pref}[1]{\ms{Pref}(#1)}
\newcommand{\Inf}{\ms{Inf}}

\title{Higher-Dimensional Automata : Extension to Infinite Tracks}
\author{Luc Passemard}{IRIF \& Université Paris Cité}{}{}{
	This work was done while the first author was a Masters intern at EPITA}
\author{Amazigh Amrane }{EPITA Research Laboratory (LRE), France}{}{}{}
\author{Uli Fahrenberg }{EPITA Research Laboratory (LRE), France}{}{}{}

\authorrunning{L.~Passemard, A.~Amrane and U.~Fahrenberg}

\Copyright{Luc Passemard, Amazigh Amrane and Uli Fahrenberg}

\ccsdesc[500]{Theory of computation~Automata over infinite objects}
\keywords{Higher-dimensional automata, concurrency theory, omega pomsets, Büchi acceptance, Muller acceptance, interval pomsets, pomsets with interfaces}

\begin{document}
	
\maketitle
\begin{abstract}
  We introduce higher-dimensional automata for infinite interval ipomsets ($\omega$-HDAs).
  We define key concepts from different points of view, inspired from their finite counterparts.
  Then we explore languages recognized by $\omega$-HDAs under Büchi and Muller semantics.
  We show that Muller acceptance is more expressive than Büchi acceptance and, in contrast to the finite case,  both semantics do not yield languages closed under subsumption.
  Then, we adapt the original rational operations to deal with $\omega$-HDAs and show that while languages of $\omega$-HDAs are $\omega$-rational, not all $\omega$-rational languages can be expressed by $\omega$-HDAs.
  
\end{abstract}
\section{Introduction}
\label{sec:introduction}


Automata theory is fundamental for modeling and analyzing computational systems. 
It is used to verify system correctness, infer models for unknown systems, synthesize components from specifications, and develop decision procedures.
Finite automata over words (Kleene automata) model terminating sequential systems with finite memory, where accepted words represent execution sequences. Their theory, backed by the Kleene, Büchi, and Myhill-Nerode theorems, connects regular expressions, monadic second-order logic, and semigroups.
For concurrent systems, executions may be represented as \emph{pomsets} (partially ordered multisets) \cite{Pratt86pomsets} instead of words. 
In a pomset, concurrent events are represented as labeled elements that are not ordered relative to each other. 
Different classes of pomsets and their associated automata models exist, reflecting diverse interpretations of concurrency. 
We can cite for example  branching automata and series-parallel pomsets~\cite{lodaya98kleene, LW98:Algebra, LW00:sp, lodaya01kleene}, step transition systems  and local trace languages~\cite{DBLP:conf/apn/FanchonM09}, communicating finite-state machines and message sequence charts \cite{GKM06}, asynchronous automata and Mazurkiewicz traces~\cite{Zielonka87} or higher-dimensional automata (HDAs) and interval pomsets~\cite{DBLP:journals/mscs/FahrenbergJSZ21}.

In this paper, we focus on HDAs \cite{Pratt91-geometry,Glabbeek91-hda}. They are general models of concurrency that extend, for example, event structures and safe Petri nets~\cite{DBLP:journals/tcs/Glabbeek06,amrane2025petrinetshigherdimensionalautomata}, asynchronous transition systems~\cite{Bednarczyk87-async, WinskelN95-Models} and obviously Kleene automata.
Initially studied from a geometrical or categorical point of view, the language theory of HDAs has become another focus for research in the past few years since \cite{DBLP:journals/mscs/FahrenbergJSZ21}. Fahrenberg et al have now shown a Kleene theorem \cite{DBLP:conf/concur/FahrenbergJSZ22}, a Myhill-Nerode theorem \cite{DBLP:journals/corr/abs-2210-08298} and a Büchi theorem \cite{amrane2024logic}. Higher-dimensional timed automata are introduced in \cite{DBLP:conf/adhs/Fahrenberg18} and their languages in \cite{amrane2024languages}.
HDAs consist of a collection of cells in which events are running concurrently, connected by face maps which model the start and termination of events.
The language of an HDA is defined as a set of \emph{interval pomsets} \cite{journals/mpsy/Fishburn70} with interfaces (interval ipomsets or \emph{iipomsets})~\cite{DBLP:journals/iandc/FahrenbergJSZ22}. 
Interval pomsets are suitable for situations where events in concurrent systems extend over time, such as producer-consumer systems, which series-parallel pomsets cannot capture.
Modelling executions with interval pomsets supports partial-order reduction, with a representation that is exponentially smaller than other alternatives.

The idea in an HDA is that each event in an execution $P$ is a time interval of process activity.
The execution is built by joining elementary steps, each representing segments of $P$. 
This gluing composition allows events to extend across segments, connecting one part to the next.
In addition, any order extension of $P$ is also a valid behaviour for the HDA. We say that the language  is closed under \emph{subsumption}.
As an example, Fig.~\ref{fig:int-conc} shows Petri net and HDA models for a system with two events, labeled $a$ and $b$. The Petri net and HDA on the left side model the (mutually exclusive) interleaving of $a$ and $b$ as either $a. b$ or $b. a$; those to the right model concurrent execution of $a$ and $b$ where the process $a \parallel b$ is a continuous path (called \emph{track}) through the surface of the filled-in square, starting at the top and terminating at the bottom node.  The shape of such a track defines the interval scheduling of $a$ and $b$ where the intervals overlap.
\begin{figure}[!h]
    \centering
    \begin{tikzpicture}
        \begin{scope}[x=1.5cm, state/.style={shape=circle, draw,
                fill=white, initial text=, inner sep=1mm, minimum size=3mm}]
            \node[state, black] (10) at (0,0) {};
            \node[state, rectangle] (20) at (0,-1) {$\vphantom{b}a$};
            \node[state] (30) at (0,-2) {};
            \node[state, black] (11) at (1,0) {};
            \node[state, rectangle] (21) at (1,-1) {$b$};
            \node[state] (31) at (1,-2) {};
            \path (10) edge (20);
            \path (20) edge (30);
            \path (11) edge (21);
            \path (21) edge (31);
            \node[state, black] (m) at (.5,-1) {};
            \path (20) edge[out=15, in=165] (m);
            \path (m) edge[out=-165, in=-15] (20);
            \path (21) edge[out=165, in=15] (m);
            \path (m) edge[out=-15, in=-165] (21);
        \end{scope}
        \begin{scope}[xshift=4cm]
            \node[state] (00) at (0,0) {};
            \node[state] (10) at (-1,-1) {};
            \node[state] (01) at (1,-1) {};
            \node[state] (11) at (0,-2) {};
            \path (00) edge node[left] {$\vphantom{b}a$\,} (10);
            \path (00) edge node[right] {\,$b$} (01);
            \path (10) edge node[left] {$b$\,} (11);
            \path (01) edge node[right] {\,$\vphantom{b}a$} (11);
        \end{scope}
    \end{tikzpicture}
    \qquad\qquad
    \begin{tikzpicture}
        \begin{scope}[xshift=8cm]
            \path[fill=black!15] (0,0) to (-1,-1) to (0,-2) to (1,-1);
            \node[state] (00) at (0,0) {};
            \node[state] (10) at (-1,-1) {};
            \node[state] (01) at (1,-1) {};
            \node[state] (11) at (0,-2) {};
            \path (00) edge node[left] {$\vphantom{b}a$\,} (10);
            \path (00) edge node[right] {\,$b$} (01);
            \path (10) edge node[left] {$b$\,} (11);
            \path (01) edge node[right] {\,$\vphantom{b}a$} (11);
        \end{scope}
        \begin{scope}[x=1.5cm, state/.style={shape=circle, draw,
                fill=white, initial text=, inner sep=1mm, minimum size=3mm},
            xshift=10.5cm]
            \node[state, black] (10) at (0,0) {};
            \node[state, rectangle] (20) at (0,-1) {$\vphantom{b}a$};
            \node[state] (30) at (0,-2) {};
            \node[state, black] (11) at (1,0) {};
            \node[state, rectangle] (21) at (1,-1) {$b$};
            \node[state] (31) at (1,-2) {};
            \path (10) edge (20);
            \path (20) edge (30);
            \path (11) edge (21);
            \path (21) edge (31);
        \end{scope}
    \end{tikzpicture}

\caption{Petri net and HDA models distinguishing interleaving (left) from non-interleaving (right) concurrency. Left: models for~$a. b+ b. a$; right: models for~$a\parallel b$.}
\label{fig:int-conc}
\end{figure}
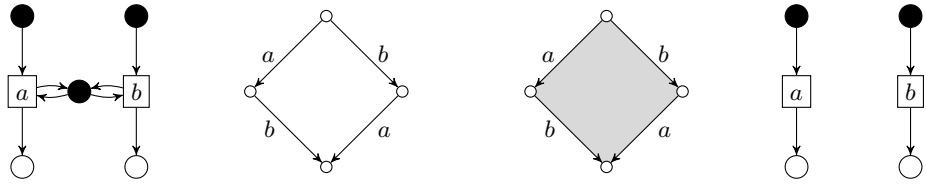

Muller \cite{muller1963infinite} and Büchi \cite{buchi1966symposium} introduced automata recognizing $\omega$-words to study non-terminating sequential machines and decision problems. McNaughton \cite{mcnaughton1966testing} later proved the equivalence of these definitions and extended Kleene’s theorem to $\omega$-words through a non-nested $\omega$-iteration. 
These automata also have logical \cite{buchi1966symposium} and algebraic \cite{carton2008automata} characterizations. Beyond their theoretical significance, they and their variations (such as Rabin, Streett, and parity automata) play a key role in specifying and verifying reactive systems \cite{duret.22.cav}. 
This led to the extension of automata models for concurrency to the infinite case, with fundamental results carrying over.  
For instance, $\omega$-branching automata admit Kleene-like and Büchi-like theorems \cite{Kuske00}, while similar developments apply to traces \cite{ebinger1996logical,diekert1993deterministic,gastin1994extension}, leading to decision procedures as corollaries. 

In this paper we introduce higher-dimensional automata for $\omega$-interval pomsets: \emph{$\omega$-HDAs}. 
To do so, we first define a class of $\omega$-interval pomsets with interfaces suitable for $\omega$-HDAs and extend the fundamental concepts to the infinite case. 
This may be effective for modelling (and checking) reactive concurrent systems that may not terminate, especially when reasoning about liveness properties.
As in the finite case, we show that isomorphisms of $\omega$-ipomsets are unique and that interval $\omega$-ipomsets admit several equivalent definitions and canonical decompositions. Then we define languages of $\omega$-HDAs in terms of interval $\omega$-ipomsets by considering Büchi and Muller acceptances. We show that unlike HDAs, languages of $\omega$-HDAs may not be closed under subsumption. 
Thus, we revise the rational operations defined in \cite{DBLP:conf/concur/FahrenbergJSZ22} by ensuring that subsumption is not implicitly assumed and introduce an omega iteration to define $\omega$-rational languages of iipomsets. We show by translating $\omega$-HDAs to $\omega$-automata over discrete ipomsets, called ST-automata, that languages of $\omega$-HDAs are $\omega$-rational.
On the other hand, we show that, unlike $\omega$-automata, Büchi acceptance is less expressive than the Muller one in $\omega$-HDAs. In addition, there are $\omega$-rational languages that cannot be recognized by Muller $\omega$-automata. To address this, we initiate first steps towards characterizing the subclass of omega-rational languages that are as expressive as Büchi or Muller $\omega$-HDAs.
We refer to the long version \cite{DBLP:journals/corr/abs-2503-07881} for detailed proofs of our results.

\section{Ipomsets}
\label{sec:infinite-ipomsets}
Ipomsets generalize pomsets \cite{DBLP:journals/fuin/Grabowski81,Pratt86pomsets, gischer1988equational}: a well-established model for non-interleaving concurrency.
In this section, we review the fundamental definitions and extend them to the infinite case.

\subsection{Finite ipomsets}
\label{sec:finite-ipomsets}

We fix a finite alphabet $\Sigma$.
An \emph{ipomset} (over $\Sigma$) is a structure $(P, {<}, {\evord}, S, T, \lambda)$ consisting of  a finite set   of \emph{events} $P$, two strict partial orders: $<$ the \emph{precedence order}, and $\evord$ the \emph{event order}, a set $S \subseteq P$ of \emph{source} interfaces, a set $T\subseteq P$ of \emph{target} interfaces such that the elements of $S$ are $<$-minimal and those of $T$ are $<$-maximal, and a labeling function  $\lambda : P \to \Sigma$.
In addition, we require that the relation ${<} \cup {\evord}$ is total, but not necessarily an order. In fact, if we have $x<y$ and $y \evord x$, then ${<} \cup {\evord}$ is not an order.

For the purpose of notation, we use $\ilo{S}{P}{T}$ for an ipomset $(P, <, \evord, S, T, \lambda)$, or refer to it by its set of events $P$ and to its components by adding the subscript $_P$.

We highlight two special cases of ipomsets: \emph{conclists} $(U, \evord, \lambda)$ where $S_U = T_U = \emptyset$ and $<_U = \emptyset$ hence $\evord$ total, with $\square$ the set of conclists, and \emph{identities} $\id_U = \ilo{U}{U}{U} $ where $S_U = T_U = U$ and $<_U = \emptyset$, with $\Id$ the set of identities.
We call $\id_{\emptyset}$  the empty conclist/identity. 
Note that for any ipomset $P$, $S_P=(S,  {\evord\rest{S\times S}}, \lambda\rest{S})$
and $T_P=(T, {\evord\rest{T\times T}}, \lambda\rest{T})$, where ``${}\rest{}$'' denotes restriction, are conclists.
An ipomset $P$ is said to be \emph{interval} (called \emph{iipomset}) if $<_P$ is an interval order, i.e.\ if for all $w,x,y,z \in P$, if $w < y$ and $x < z$, then $w < z$ or $x < y$.

Let $P$ and $Q$ be two ipomsets. We say that $Q$ \emph{subsumes} $P$ (written $P \subsu Q$) if there is a bijection $f: P \to Q$ (a \emph{subsumption}) such that
$\lambda_Q \circ f = \lambda_P$, $f(S_P) = S_Q$, and $f(T_P) = T_Q$,
$\forall x,y \in P,\, f(x) <_Q f(y) \implies x <_P y$ and
$\forall x,y \in P$ such that $x \nless_P y$ and $y \nless_P x$, we have $x \evord_P y \implies f(x) \evord_Q f(y)$.
Informally speaking, $P$ is more precedence ordered than $Q$.

\begin{example}
 An example of subsumption is depicted in Fig.~\ref{fig:interval-ipomset}.
 Note that  $e_1< e_2$ when the activity interval of $e_1$ finishes before the beginning of the one of $e_2$.
 When the activity intervals of two events overlap, they are $\evord$-ordered from top to bottom.
\end{example}

An \emph{isomorphism} of ipomsets is an invertible subsumption (whose inverse is again a subsumption);
isomorphic ipomsets are denoted $P \cong Q$.
Due to the totality of $< \cup \evord$,  isomorphisms between ipomsets are unique~\cite{PIPI},
so we may switch freely between ipomsets and their isomorphism classes.
The set of ipomsets is denoted $\ipoms$, and the set of interval ipomsets is written $\iipoms$. 
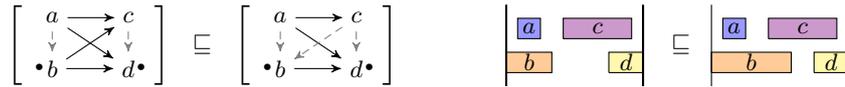
\begin{figure}[!h]

    \centering
    \begin{tikzpicture}
    
        \def\possh{-1.3};
        \def\hw{0.3};
        \begin{scope}[shift={(0,0)}]
        \node at (-2, 0) {%
	$\left[\vcenter{\hbox{\!%
			\begin{tikzpicture}[x=1cm,y=0.9cm]
				\node (a) at (0,0) {$a$};
				\node (c) at (1,0) {$c$};
				\node at (1.17,-0.75) {$\ibu$};
				\node (b) at (0,-.75) {$b$};
				\node at (-.17,-.75) {$\ibu$};
				\node (d) at (1,-.75) {$d$};
				\path (a) edge (c);
				\path (b) edge (d);
				\path (a) edge (d);
				\path (b) edge (c);
				\path[densely dashed, gray] (a) edge (b)  (c) edge (d);
			\end{tikzpicture}
			\!\!}}\right]$};     
		
        \node at (1, 0) {%
        	$\left[\vcenter{\hbox{\!%
      			\begin{tikzpicture}[x=1cm,y=0.9cm]
      				\node (a) at (0,0) {$a$};
      				\node (c) at (1,0) {$c$};
      				\node at (1.17,-0.75) {$\ibu$};
      				\node (b) at (0,-.75) {$b$};
      				\node at (-.17,-.75) {$\ibu$};
      				\node (d) at (1,-.75) {$d$};
      				\path (a) edge (c);
      				\path (b) edge (d);
      				\path (a) edge (d);
      				\path[densely dashed, gray] (a) edge (b) (c) edge (b) (c) edge (d);
      			\end{tikzpicture}
      			\!\!}}\right]$};     
         \end{scope}
     	\node at (-0.5,0) {$\subsu$};                
        \begin{scope}[shift={(3.5,-0.55)},x=0.5cm,y=0.9cm]
            \filldraw[fill=blue!40!white,-](0.3,0.7)--(0.9,0.7)--(0.9,0.7+\hw)--(0.3,0.7+\hw)--(0.3,0.7);
            \filldraw[fill=orange!40!white,-](0.0,0.2)--(1.2,0.2)--(1.2,0.2+\hw)--(0.0,0.2+\hw)--(0.0,0.2);
            \filldraw[fill=violet!40!white,-](1.5,0.7)--(3.3,0.7)--(3.3,0.7+\hw)--(1.5,0.7+\hw)--(1.5,0.7);
            \filldraw[fill=yellow!40!white,-](2.7,0.2)--(3.6,0.2)--(3.6,0.2+\hw)--(2.7,0.2+\hw)--(2.7,0.2);
            \draw[thick,-](0,0)--(0,1.2);
            \draw[thick,-](3.6,0)--(3.6,1.2);
            \node at (0.6,0.7+\hw*0.5) {$a$};
            \node at (0.6,0.2+\hw*0.5) {$b$};
            \node at (2.4,0.7+\hw*0.5) {$c$};
            \node at (3.15,0.2+\hw*0.5) {$d$};
        \end{scope}
         	\node at (5.8,0) {$\subsu$};                
      \begin{scope}[shift={(6.2,-0.55)},x=0.5cm,y=0.9cm]
			\filldraw[fill=blue!40!white,-](0.3,0.7)--(0.9,0.7)--(0.9,0.7+\hw)--(0.3,0.7+\hw)--(0.3,0.7);
			\filldraw[fill=orange!40!white,-](0.0,0.2)--(2.1,0.2)--(2.1,0.2+\hw)--(0.0,0.2+\hw)--(0.0,0.2);
			\filldraw[fill=violet!40!white,-](1.5,0.7)--(3.3,0.7)--(3.3,0.7+\hw)--(1.5,0.7+\hw)--(1.5,0.7);
			\filldraw[fill=yellow!40!white,-](2.7,0.2)--(3.6,0.2)--(3.6,0.2+\hw)--(2.7,0.2+\hw)--(2.7,0.2);
			\draw[-](0,0)--(0,1.2);
			\draw[-](3.6,0)--(3.6,1.2);
			\node at (0.6,0.7+\hw*0.5) {$a$};
			\node at (1.05,0.2+\hw*0.5) {$b$};
			\node at (2.4,0.7+\hw*0.5) {$c$};
			\node at (3.15,0.2+\hw*0.5) {$d$};
		\end{scope}

    \end{tikzpicture}
    
\caption{Ipomsets (left) and their corresponding interval representations (right).
	Full arrows indicate precedence order; dashed arrows indicate event order; bullets indicate interfaces.
}
\label{fig:interval-ipomset}
\end{figure}

In an ipomset $P$, a \emph{chain} is a subset of $P$ totally ordered by $<_P$.
An \emph{antichain} $A$ of $P$ is such that ${<_P} \cap (A \times A) = \emptyset$. Hence $A$ is totally ordered by $\evord_{P}$.
The width of $P$ is ${\ms{wd}(P) = \sup{|A| \midbar A \text{ is an $<$-antichain of }P}}$.
\subsection{The set of $\omega$-ipomsets}
\label{sec:w-ipomset}

Similarly to ipomsets, an $\omega$-ipomset is a structure of the form $(P,{<},{\evord},S,\lambda)$.
The only changes are that $P$ is countably infinite and there is no target interface. 
However, let $\Omega_P$ denote the set of \emph{all} $<_P$-maximal elements.
This is not similar to a target interface: even if $\Omega_{P}$ represents unfinished events, terminating them (by trying to glue an ipomset after $P$) may cause chains greater than $\omega$.

For the remainder of the paper, we  focus on a subclass of $\omega$-ipomsets.
\begin{definition}[Valid $\omega$-ipomset]
\label{def:validity}
    A \emph{valid} $\omega$-ipomset is an $\omega$-ipomset $P$ such
    that
    \begin{multicols}{2}
	    \begin{itemize}
		\item
		$\forall x \in P,\, \set{ y \midbar y <_P x}$ is finite;
		\item
		every $<_P$-antichain of $P$ is finite.
		\end{itemize}
    \end{multicols}
\end{definition}
The first point is here to forbid $\omega$-ipomsets with chains greater than $\omega$, and therefore extend properly classical $\omega$-words and $\omega$-pomsets~\cite{Kuske00}. It is a classical property required in event structures (see~\cite{Winskel87} for example).
The second point is to avoid an infinite number of concurrent events. 
This condition implies in particular that $S_P$ and $\Omega_P$ are finite.
Both conditions also imply that $<_P$ is a well quasi-order, but the converse is not true.
In the rest of this paper, unless specified, we will talk about valid $\omega$-ipomsets, and will omit the word ``valid''.


Let $P$ and $Q$ be two (valid) $\omega$-ipomsets. 
We define the subsumption $P \subsu Q$ as in the finite case by preserving also the $<_P$-maximal elements $\Omega_P$. 
Again, isomorphisms of $\omega$-ipomsets are invertible subsumptions (whose inverse is again a subsumption). 
The first result of this paper is the unicity of these isomorphisms. The proof is similar to the finite case, using pomset filtration (see \cite[Lem.~34]{DBLP:journals/mscs/FahrenbergJSZ21}):

\begin{proposition}
\label{prop:unicity-iso}
Isomorphisms between $\omega$-ipomsets are unique. 
\end{proposition}

Note that the proposition does not hold for general (non-valid) $\omega$-ipomsets: for example, all shifts by $n \in \Int$ are non-trivial automorphisms of $\Int$.
As in the finite case, we may now switch freely between $\omega$-ipomsets and their isomorphism classes.
The set of $\omega$-ipomsets is denoted $\wipoms$, and we use the notation $\infipoms \coloneqq \ipoms \cup \wipoms$.

\subsection{Operations on $\omega$-ipomsets}
\label{sec:operation}

We extend here the operations $*$ (gluing) and $\parallel$ (parallel composition) to $\omega$-ipomsets. 

\label{def:parallel-composition}
For $(P,Q) \in (\infipoms \times \infipoms) \setminus(\ipoms \times \ipoms)$, the \emph{parallel composition} $P\parallel Q$ (or $\loset{P\\Q}$) is the $\omega$-ipomset $(P \sqcup Q,\, <_P \sqcup <_Q,\, \evord,\, S_P \sqcup S_Q,\, \lambda_P \sqcup \lambda_Q)$ where $\sqcup$ is the disjoint union and  $x\evord y$ iff  $(x,y)\in P \times Q$ or $x\evord_P y$ or $x\evord_Q y$.

The \emph{gluing composition} $P*Q$ (or $PQ$ ) of $(P,Q) \in \ipoms \times \wipoms$ is defined if there exists a unique\footnote{Isomorphisms of conclists is a special case of ipomset isomorphisms. Unicity is ensured by event order.} isomorphism $f \colon T_P \to S_Q$, by $((P\sqcup Q)_{/x\sim f(x)},<,\evord,S_P,\lambda_P \cup \lambda_Q)$\footnote{($P\sqcup Q)_{x\sim f(x)}$ is the quotient of the disjoint union under the unique isomorphism $f: T_P\to S_Q$.}, where:
	\begin{itemize}
            \item $x< y$ iff $x<_P y$, $x<_Q y$, or $x\in P\setminus T_P$ and $y\in Q\setminus S_Q$;
        \item
            $\evord$ is the transitive closure of $\evord_P \cup \evord_Q$ on $(P\sqcup Q)_{/x\sim f(x)}$.
    \end{itemize}
Informally, the sources of $Q$ are attached to the targets of $P$.

Note that both $*$ and $\parallel$ are associative and non-commutative.
The non-commutativity of $\parallel$ is due to event order. 
The latter is crucial for  isomorphism unicity  and to define gluing composition.
Still, even	 if $\loset{a \\ b} \neq \loset{b \\ a}$, both express that $a$ and $b$ are running concurrently.
Alur et al in \cite{alur2023robust} introduced a similar “ordered parallel composition” for application purposes.

We also introduce an \emph{infinite gluing} of an $\omega$-sequence of ipomsets.
    Let $(P_i)_{i \in \Nat} \in \ipoms^\Nat$ such that for all $i$, $T_{P_i} \cong S_{P_{i\+1}}$. We define $P = P_0 * P_1 * \cdots$ by $(P,<,\evord,S_P,\lambda_P)$, where
    \begin{itemize}
        \item 
            $P = (\bigcup\nolimits_{i \in \Nat} (P_i,i) )_{/x \sim f(x)}$  where $f :\bigcup\nolimits_{i \in \Nat} (T_i,i) \rightarrow \bigcup\nolimits_{i \in \Nat} (S_i,i)$ such that $f(x,i) = (f_i(x),$ $i+1)$ with $f_i$ the unique isomorphism between $T_{P_i}$ and $S_{P_{i\+1}}$;
        \item
            $S_P = S_{P_0} \times \set{0}$;
        \item
             $(x,i) <_P (y,j)$ if  
             ($i = j$ and $x <_{P_i} y$) or 
              ($i\+1 \leq j$, $x \in P_i \setminus T_{P_i}$ and $y \in P_{i\+1} \setminus S_{P_{i\+1}}$);
        \item
             $\evord_P$ is the transitive closure of $ \bigcup_{i \in \Nat} \evord_{P_i}$; 
        \item
             $\lambda_P (x,i) = \lambda_{P_i}(x)$ for $i \in \Nat$ and $x \in P_i$.
    \end{itemize}

Note in particular that when $P_i \in \Id$ we have $P_0* \dots *P_{i-1} * P_i * P_{i+1} *\dots =P_0 * \dots *P_{i-1} * P_{i+1} *\cdots$. 
We let denote by $P^\omega$ the $\omega$-ipomset defined by the infinite gluing of the constant sequence equal to $P$ with $S_P \cong T_P$.

\begin{lemma}
\label{lem:infinite-product-well-defined}
    For $(P_i)_{i \in \Nat} \in \ipoms^\Nat$ with $T_{P_i} \cong S_{P_{i\+1}}$,  $P = P_0 * P_1 * \cdots$ is valid if the number of $P_i \in \ipoms \setminus \Id$ is infinite and there exists $m \in \Nat$ such that for all $i \in \Nat, ~ \ms{wd}(P_i)\leq m$. 
\end{lemma}
Intuitively, $P$ is infinite since it is obtained by composing infinitely many non-identities, finite past is a consequence of the infinite gluing definition, and finiteness of antichains is due to width-boundedness of the operands.
Note that this is not an equivalence condition. 
The infinite gluing may define a valid $\omega$-ipomset even if $\ms{wd}(P_i)$ are not bounded (see Fig.~\ref{fig:infinite-product-valid}). Finding an equivalence condition on $(P_i)$ is hard, because it has to consider event order. 
For example, Fig.~\ref{fig:infinite-product-evord} exhibits two infinite products behaving differently with only a change of event order.

\begin{figure}[!h]
  
    \centering
    \begin{tikzpicture}[y=1.1cm]
        
        \def\hw{0.3}; 
        \def\shift{0.55}; 
        
        \begin{scope}[shift={(0,1.5)}]

            \node at (-0.6,0) {$P =$};
            
            \node at (0.1,0) {%
                $\left[ \vcenter{\hbox{\!%
                \begin{tikzpicture}[x=1.2cm]
                    \node (a) at (0,0) {$a$};
                \end{tikzpicture}
              \!\!}} \right]$};

            \node at (1.2,0) {%
                $\left[ \vcenter{\hbox{\!%
                \begin{tikzpicture}[x=1.2cm]
                    \node (a1) at (0,0) {$a$};
                    \node (a2) at (0,-0.25) {$a$};
                \end{tikzpicture}
              \!\!}} \right]$};
            
            \node at (2.4,0) {%
                $\left[ \vcenter{\hbox{\!%
                \begin{tikzpicture}[x=1.2cm]
                    \node (a1) at (0,0) {$a$};
                    \node (a2) at (0,-0.25) {$a$};
                    \node (a3) at (0,-0.5) {$a$};
                \end{tikzpicture}
              \!\!}} \right]$};
          
            \node at (0.6,0) {*};                
            \node at (1.8,0) {*};
            \node at (3,0) {*};
            \node at (3.4,0) {...};
            \node at (3.8,0) {$=$};
        
        \end{scope}

        \begin{scope}[shift={(0,0.1)}]

            \node at (-0.6,0) {$Q =$};
            
            \node at (0.1,0) {%
                $\left[ \vcenter{\hbox{\!%
                \begin{tikzpicture}[x=1.2cm]
                    \node (a) at (0,0) {$a \ibu$};
                \end{tikzpicture}
              \!\!}} \right]$};

            \node at (1.2,0) {%
                $\left[ \vcenter{\hbox{\!%
                \begin{tikzpicture}[x=1.2cm]
                    \node (a1) at (0,0) {$\ibu a \ibu$};
                    \node (a2) at (0,-0.25) {$\nibu a \ibu$};
                \end{tikzpicture}
              \!\!}} \right]$};
            
            \node at (2.4,0) {%
                $\left[ \vcenter{\hbox{\!%
                \begin{tikzpicture}[x=1.2cm]
                    \node (a1) at (0,0) {$\ibu a \ibu$};
                    \node (a2) at (0,-0.25) {$\ibu a \ibu$};
                    \node (a3) at (0,-0.5) {$\nibu a \ibu$};
                \end{tikzpicture}
              \!\!}} \right]$};
                    
            \node at (0.6,0) {*};                
            \node at (1.8,0) {*};
            \node at (3,0) {*};
            \node at (3.4,0) {...};
            \node at (3.8,0) {$=$};
        
        \end{scope}
       
        \begin{scope}[shift={(4.5,1.5)},y=0.7cm]
    
            \draw[thick,-](0,-0.9)--(0,0.9);

            \draw[dashed,gray,-](2 + 0*2, -0.8)--(2 + 0*2, 0.8);
            \draw[dashed,gray,-](2 + 1*2, -0.8)--(2 + 1*2, 0.8);
            \draw[dashed,gray,-](2 + 2*2, -0.8)--(2 + 2*2, 0.8);
          
            \filldraw[fill=yellow!60!white,-](0.2 + 0*2,0.4 - 0*\shift)--(1.8 + 0*2,0.4 - 0*\shift)--(1.8 + 0*2,0.4 + \hw - 0*\shift)--(0.2 + 0*2,0.4 + \hw - 0*\shift)--(0.2 + 0*2,0.4 - 0*\shift);              
            \node at (1 + 0*2, 0.4 + \hw*0.5) {$a$};

            \filldraw[fill=yellow!60!white,-](0.2 + 1*2,0.4)--(1.8 + 1*2,0.4)--(1.8 + 1*2,0.4+\hw)--(0.2 + 1*2,0.4+\hw)--(0.2 + 1*2,0.4);              
            \node at (1 + 1*2, 0.4 + \hw*0.5) {$a$};
            \filldraw[fill=yellow!60!white,-](0.2 + 1*2,0.4 - 1*\shift)--(1.8 + 1*2,0.4 - 1*\shift)--(1.8 + 1*2,0.4 + \hw - 1*\shift)--(0.2 + 1*2,0.4 + \hw - 1*\shift)--(0.2 + 1*2,0.4 - 1*\shift);            
            \node at (1 + 1*2, 0.4 + \hw*0.5 - 1*\shift) {$a$};

            \filldraw[fill=yellow!60!white,-](0.2 + 2*2,0.4)--(1.8 + 2*2,0.4)--(1.8 + 2*2,0.4+\hw)--(0.2 + 2*2,0.4+\hw)--(0.2 + 2*2,0.4);              
            \node at (1 + 2*2, 0.4 + \hw*0.5) {$a$};
            \filldraw[fill=yellow!60!white,-](0.2 + 2*2,0.4 - 1*\shift)--(1.8 + 2*2,0.4 - 1*\shift)--(1.8 + 2*2,0.4 + \hw - 1*\shift)--(0.2 + 2*2,0.4 + \hw - 1*\shift)--(0.2 + 2*2,0.4 - 1*\shift);
            \node at (1 + 2*2, 0.4 + \hw*0.5 - 1*\shift) {$a$};
            \filldraw[fill=yellow!60!white,-](0.2 + 2*2,0.4 - 2*\shift)--(1.8 + 2*2,0.4 - 2*\shift)--(1.8 + 2*2,0.4 + \hw - 2*\shift)--(0.2 + 2*2,0.4 + \hw - 2*\shift)--(0.2 + 2*2,0.4 - 2*\shift);  
            \node at (1 + 2*2, 0.4 + \hw*0.5 - 2*\shift) {$a$};
         
            \node at (6.5,0) {...};
    
        \end{scope}

        \begin{scope}[shift={(4.5,0.1)},y=0.7cm]

            \draw[thick,-](0,-0.9)--(0,0.9);

            \draw[dashed,gray,-](2 + 0*2, -0.8)--(2 + 0*2, 0.8);
            \draw[dashed,gray,-](2 + 1*2, -0.8)--(2 + 1*2, 0.8);
            \draw[dashed,gray,-](2 + 2*2, -0.8)--(2 + 2*2, 0.8);
         
            \filldraw[fill=blue!60!white,-](0.2 + 0*2, 0.4 - 0*\shift)--(6 , 0.4 - 0*\shift)--(6, 0.4 + \hw - 0*\shift)--(0.2 + 0*2,0.4 + \hw - 0*\shift)--(0.2 + 0*2,0.4 - 0*\shift);              
            \node at (1 + 0*2, 0.4 + \hw*0.5) {$a$};

            \filldraw[fill=blue!60!white,-](0.2 + 1*2,0.4 - 1*\shift)--(6, 0.4 - 1*\shift)--(6, 0.4 + \hw - 1*\shift)--(0.2 + 1*2,0.4 + \hw - 1*\shift)--(0.2 + 1*2,0.4 - 1*\shift);            
            \node at (1 + 1*2, 0.4 + \hw*0.5 - 1*\shift) {$a$};

            \filldraw[fill=blue!60!white,-](0.2 + 2*2,0.4 - 2*\shift)--(6, 0.4 - 2*\shift)--(6, 0.4 + \hw - 2*\shift)--(0.2 + 2*2,0.4 + \hw - 2*\shift)--(0.2 + 2*2,0.4 - 2*\shift);  
            \node at (1 + 2*2, 0.4 + \hw*0.5 - 2*\shift) {$a$};

            \node at (6.5,0) {...};  
    
        \end{scope}

  \end{tikzpicture}

\caption{Two unbounded infinite products giving a valid (up) and an invalid (down) $\omega$-ipomset.}
\label{fig:infinite-product-valid}
\end{figure}
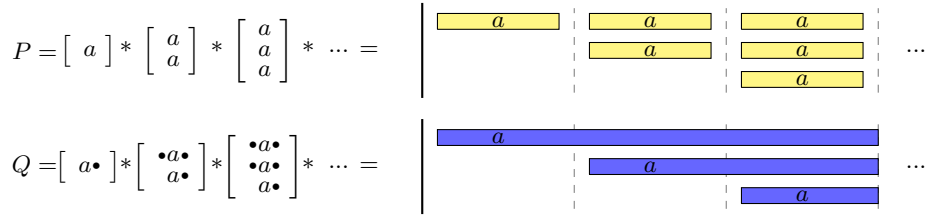

\subsection{Interval $\omega$-ipomsets}
\label{sec:w-iipoms}

We will here only deal with interval $\omega$-ipomsets.
An $\omega$-ipomset is \emph{interval} (denoted $\omega$-iipomset) if
for all $w,x,y,z \in P$, if $w < y$ and $x < z$, then $w < z$ or $x < y$.
The set of all interval $\omega$-ipomsets is denoted $\wiipoms$ (and we use the notation $\infiipoms \coloneqq \iipoms \cup \wiipoms$). 

There are multiple equivalent ways to define interval ipomsets~\cite[Lem.~39]{DBLP:journals/mscs/FahrenbergJSZ21}, and the same goes for $\omega$-ipomsets. 

\label{def:interval-repr}
We say that $P \in \wipoms$ admits an \emph{interval representation} if there are two functions ${b,e: P \rightarrow \Nat \cup \set{+\infty}}$ such that for all $x,y \in P$, $b(x) \leq e(x)$ and $x <_P y \equivalent e(x) < b(y)$.

As for finite pomsets and with similar arguments, the following holds:
\begin{proposition}
\label{prop:eq-wiipoms}
    Let $P$ be an $\omega$-ipomset, the following are equivalent:
    \begin{enumerate}
        \item 
            $P$ is an interval $\omega$-ipomset;
        \item
            $P$ has an interval representation;
        \item
            the order $\ll$ on maximal antichains of $P$ defined by $A \ll B$ iff $A \neq B$ and $ \forall (a,b) \in A \times B$,\, $b \nless_P a$ is linear.
    \end{enumerate}
\end{proposition}

Establishing the equivalence between the first two items is a routine.
The third item is less intuitive, but it implies the second one by showing that the set of maximal antichains equipped with $\ll$ is isomorphic to $(\Nat,<)$.
From that, to an event $a \in P$ we can associate the interval going from the index of the first appearance of $a$ in the antichains to the index of the last appearance.

\begin{example}
\label{ex:interval-wipoms}
The minimal non-interval ipomset $aa \parallel bb$ can be generalized to the non-interval $\omega$-ipomset $a^\omega \parallel b^\omega$, while $(a \parallel b)^\omega \in \wiipoms$, as well as $\loset{ \ibu a a \ibu \\ b}^\omega$ represented in Fig.~\ref{fig:interval-example}.
\end{example}

\begin{figure}[!h]
    \centering
    \begin{tikzpicture}[y=0.9cm]
    
        \def\hw{0.3}; 

        
        \node at (1,2) {%
            $\left[ \vcenter{\hbox{\!%
            \begin{tikzpicture}[x=1.2cm]
                \node (a) at (0,0) {$a$};
                \node at (-.17,0) {$\ibu$};
                \node (a') at (1,0) {$a$};
                \node at (1.17,0) {$\ibu$};
                \node (b) at (0.5,-.75) {$b$};
                \path (a) edge (a');
                \path[densely dashed, gray] (a) edge (b) (a') edge (b);
            \end{tikzpicture}
          \!\!}} \right]^\omega$};
    
        \node at (3,2) {$=$};

        \begin{scope}[shift={(3.5,1.4)}]
            \draw[thick,-](0.4,0.1)--(0.4,1.2);

            \draw[dashed,gray,-](2.2 + 0*1.6,0.1)--(2.2 + 0*1.6,1.2);
            \filldraw[fill=blue!60!white,-](0.4 + 0*1.6,0.7)--(1.2 + 0*1.6,0.7)--(1.2 + 0*1.6,0.7+\hw)--(0.4 + 0*1.6,0.7+\hw)--(0.4 + 0*1.6,0.7);
            \filldraw[fill=yellow!80!white,-](0.8 + 0*1.6,0.2)--(2.0 + 0*1.6,0.2)--(2.0 + 0*1.6,0.2+\hw)--(0.8 + 0*1.6,0.2+\hw)--(0.8 + 0*1.6,0.2);                
            \node at (0.6 + 0*1.6,0.7+\hw*0.5) {$a_1$};
            \node at (1.4 + 0*1.6,0.2+\hw*0.5) {$b_1$};

            \draw[dashed,gray,-](2.2 + 1*1.6,0.1)--(2.2 + 1*1.6,1.2);            
            \filldraw[fill=blue!60!white,-](0.0 + 1*1.6,0.7)--(1.2 + 1*1.6,0.7)--(1.2 + 1*1.6,0.7+\hw)--(0.0 + 1*1.6,0.7+\hw)--(1*1.6,0.7);
            \filldraw[fill=yellow!80!white,-](0.8 + 1*1.6,0.2)--(2.0 + 1*1.6,0.2)--(2.0 + 1*1.6,0.2+\hw)--(0.8 + 1*1.6,0.2+\hw)--(0.8 + 1*1.6,0.2);
            \node at (0.6 + 1*1.6,0.7+\hw*0.5) {$a_2$};
            \node at (1.4 + 1*1.6,0.2+\hw*0.5) {$b_2$};

            \draw[dashed,gray,-](2.2 + 2*1.6,0.1)--(2.2 + 2*1.6,1.2);            
            \filldraw[fill=blue!60!white,-](0.0 + 2*1.6,0.7)--(1.2 + 2*1.6,0.7)--(1.2 + 2*1.6,0.7+\hw)--(0.0 + 2*1.6,0.7+\hw)--(0.0 + 2*1.6,0.7);
            \filldraw[fill=yellow!80!white,-](0.8 + 2*1.6,0.2)--(2.0 + 2*1.6,0.2)--(2.0 + 2*1.6,0.2+\hw)--(0.8 + 2*1.6,0.2+\hw)--(0.8 + 2*1.6,0.2);
            \node at (0.6 + 2*1.6,0.7+\hw*0.5) {$a_3$};
            \node at (1.4 + 2*1.6,0.2+\hw*0.5) {$b_3$};

            \draw[dashed,gray,-](2.2 + 3*1.6,0.1)--(2.2 + 3*1.6,1.2);            
            \filldraw[fill=blue!60!white,-](0.0 + 3*1.6,0.7)--(1.2 + 3*1.6,0.7)--(1.2 + 3*1.6,0.7+\hw)--(0.0 + 3*1.6,0.7+\hw)--(0.0 + 3*1.6,0.7);
            \filldraw[fill=yellow!80!white,-](0.8 + 3*1.6,0.2)--(2.0 + 3*1.6,0.2)--(2.0 + 3*1.6,0.2+\hw)--(0.8 + 3*1.6,0.2+\hw)--(0.8 + 3*1.6,0.2);
            \node at (0.6 + 3*1.6,0.7+\hw*0.5) {$a_4$};
            \node at (1.4 + 3*1.6,0.2+\hw*0.5) {$b_4$};
            \draw[dashed,gray,-](2.2 + 3*1.6,0.1)--(2.2 + 3*1.6,1.2);

            \draw[dashed,gray,-](2.2 + 4*1.6,0.1)--(2.2 + 4*1.6,1.2);
            \filldraw[fill=blue!60!white,-](0.0 + 4*1.6,0.7)--(1.2 + 4*1.6,0.7)--(1.2 + 4*1.6,0.7+\hw)--(0.0 + 4*1.6,0.7+\hw)--(0.0 + 4*1.6,0.7);
            \filldraw[fill=yellow!80!white,-](0.8 + 4*1.6,0.2)--(2.0 + 4*1.6,0.2)--(2.0 + 4*1.6,0.2+\hw)--(0.8 + 4*1.6,0.2+\hw)--(0.8 + 4*1.6,0.2);
            \node at (0.6 + 4*1.6,0.7+\hw*0.5) {$a_5$};
            \node at (1.4 + 4*1.6,0.2+\hw*0.5) {$b_5$};
   
            \filldraw[fill=blue!60!white,-](0.0 + 5*1.6,0.7)--(0.8 + 5.0*1.6,0.7)--(0.8 + 5.0*1.6,0.7+\hw)--(0.0 + 5*1.6,0.7+\hw)--(0.0 + 5*1.6,0.7);
            \node at (0.6 + 5*1.6,0.7+\hw*0.5) {$a_6$};
            
            \node at (9.3,0.45+\hw*0.5) {$\dots$};
        \end{scope}
    \end{tikzpicture}

    \vspace*{-.85cm}
    \begin{align*}
        \set{a_1,b_1} \ll \set{b_1,a_2} \ll \set{a_2,b_2}& \ll \set{b_2,a_3} \ll \set{a_3,b_3} \ll \set{b_3,a_4} \ll \set{a_4,b_4} \ll ...
    \end{align*}
\caption{Interval representation and linear order on antichains of the $\omega$-iipomset of Ex.~\ref{ex:interval-wipoms}.}
\label{fig:interval-example}
\end{figure}
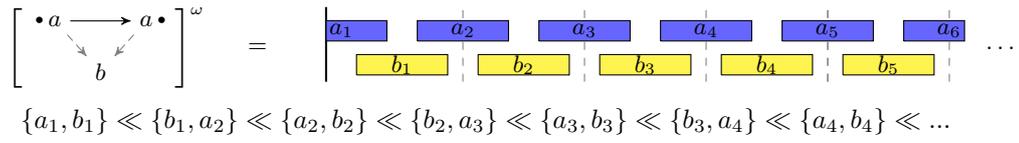

\subsection{Decomposition of $\omega$-iipomsets}
\label{sec:dec-iipoms}

A \emph{starter} $U$ is a discrete ipomset (i.e.\ $<_U$ is empty) with $T_U = U$ (hence written $\ilo{S}{U}{U}$), and a \emph{terminator} $U$ is a discrete ipomset with $S_U = U$ (written $\ilo{U}{U}{T}$). We denote by $\Xi$ the set of starters and terminators. 
Note that $\Id \subseteq \Xi$.
We say that $U \in \Xi$ is \emph{proper} if $U \in \Xi \setminus \Id$. A starter $\ilo{U-A}{U}{U}$ will be written $\starter{U}{A}$ (meaning it contains the events in $U$ and start the events in $A \subset U$), and a terminator $\ilo{U}{U}{U-A}$ will be written $\terminator{U}{A}$. We call a \emph{step decomposition} of $P \in \wipoms$ a sequence of starters and terminators $(U_i)_{i \in \Nat}$ such that $P = U_0 * U_1 * \cdots$, and a decomposition is said to be \emph{sparse} if proper starters and terminators are alternating.
As for finite iipomsets~\cite{DBLP:journals/corr/abs-2210-08298}, we can prove uniqueness of sparse decompositions of $\omega$-iipomsets:

\begin{theorem}
	\label{th:sparse-dec}
	Every interval $\omega$-ipomset has a unique sparse step decomposition.
\end{theorem}

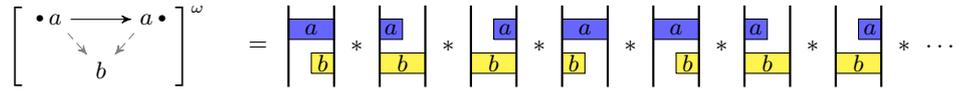
\begin{figure}[!h]
    
    \centering
    \begin{tikzpicture}[y=0.9cm]

        \def\hw{0.3}; 

        
        \node at (1,2) {%
            $\left[ \vcenter{\hbox{\!%
            \begin{tikzpicture}[x=1.2cm]
                \node (a) at (0,0) {$a$};
                \node at (-.17,0) {$\ibu$};
                \node (a') at (1,0) {$a$};
                \node at (1.17,0) {$\ibu$};
                \node (b) at (0.5,-.75) {$b$};
                \path (a) edge (a');
                \path[densely dashed, gray] (a) edge (b) (a') edge (b);
            \end{tikzpicture}
          \!\!}} \right]^\omega$};
    

        \node at (3,2) {$=$};
        \node at (4.3,2) {$*$};
        \node at (5.5,2) {$*$};
        \node at (6.7,2) {$*$};
        \node at (7.9,2) {$*$};
        \node at (9.1,2) {$*$};
        \node at (10.3,2) {$*$};
        \node at (11.5,2) {$*$};
        \node at (12,2) {$\dots$};

        \begin{scope}[shift={(3.4,1.4)}]
        
            \filldraw[fill=blue!60!white,-](0.0 + 0*1.2,0.7)--(0.6 + 0*1.2,0.7)--(0.6 + 0*1.2,0.7+\hw)--(0.0 + 0*1.2,0.7+\hw)--(0.0 + 0*1.2,0.7);
            \filldraw[fill=yellow!80!white,-](0.3 + 0*1.2,0.2)--(0.6 + 0*1.2,0.2)--(0.6 + 0*1.2,0.2+\hw)--(0.3 + 0*1.2,0.2+\hw)--(0.3 + 0*1.2,0.2);
            \draw[thick,-](0 + 0*1.2,0)--(0 + 0*1.2,1.2);
            \draw[thick,-](0.6 + 0*1.2,0)--(0.6 + 0*1.2,1.2);
            \node at (0.3 + 0.0  + 0*1.2,0.7+\hw*0.5) {$a$};
            \node at (0.3 + 0.15 + 0*1.2,0.2+\hw*0.5) {$b$};

            \filldraw[fill=blue!60!white,-](0.0 + 1*1.2,0.7)--(0.3 + 1*1.2,0.7)--(0.3 + 1*1.2,0.7+\hw)--(0.0 + 1*1.2,0.7+\hw)--(0.0 + 1*1.2,0.7);
            \filldraw[fill=yellow!80!white,-](0.0 + 1*1.2,0.2)--(0.6 + 1*1.2,0.2)--(0.6 + 1*1.2,0.2+\hw)--(0.0 + 1*1.2,0.2+\hw)--(0.0 + 1*1.2,0.2);
            \draw[thick,-](0 + 1*1.2,0)--(0 + 1*1.2,1.2);
            \draw[thick,-](0.6 + 1*1.2,0)--(0.6 + 1*1.2,1.2);
            \node at (0.3 - 0.15 + 1*1.2,0.7+\hw*0.5) {$a$};
            \node at (0.3 + 0.0   + 1*1.2,0.2+\hw*0.5) {$b$};

            \filldraw[fill=blue!60!white,-](0.3 + 2*1.2,0.7)--(0.6 + 2*1.2,0.7)--(0.6 + 2*1.2,0.7+\hw)--(0.3 + 2*1.2,0.7+\hw)--(0.3 + 2*1.2,0.7);
            \filldraw[fill=yellow!80!white,-](0.0 + 2*1.2,0.2)--(0.6 + 2*1.2,0.2)--(0.6 + 2*1.2,0.2+\hw)--(0.0 + 2*1.2,0.2+\hw)--(0.0 + 2*1.2,0.2);
            \draw[thick,-](0 + 2*1.2,0)--(0 + 2*1.2,1.2);
            \draw[thick,-](0.6 + 2*1.2,0)--(0.6 + 2*1.2,1.2);
            \node at (0.3 + 0.15 + 2*1.2,0.7+\hw*0.5) {$a$};
            \node at (0.3 + 0.0 + 2*1.2,0.2+\hw*0.5) {$b$};

            \filldraw[fill=blue!60!white,-](0.0 + 3*1.2,0.7)--(0.6 + 3*1.2,0.7)--(0.6 + 3*1.2,0.7+\hw)--(0.0 + 3*1.2,0.7+\hw)--(0.0 + 3*1.2,0.7);
            \filldraw[fill=yellow!80!white,-](0.0 + 3*1.2,0.2)--(0.3 + 3*1.2,0.2)--(0.3 + 3*1.2,0.2+\hw)--(0.0 + 3*1.2,0.2+\hw)--(0.0 + 3*1.2,0.2);
            \draw[thick,-](0 + 3*1.2,0)--(0 + 3*1.2,1.2);
            \draw[thick,-](0.6 + 3*1.2,0)--(0.6 + 3*1.2,1.2);
            \node at (0.3 + 0.0  + 3*1.2,0.7+\hw*0.5) {$a$};
            \node at (0.3 - 0.15  + 3*1.2,0.2+\hw*0.5) {$b$};

            \filldraw[fill=blue!60!white,-](0.0 + 4*1.2,0.7)--(0.6 + 4*1.2,0.7)--(0.6 + 4*1.2,0.7+\hw)--(0.0 + 4*1.2,0.7+\hw)--(0.0 + 4*1.2,0.7);
            \filldraw[fill=yellow!80!white,-](0.3 + 4*1.2,0.2)--(0.6 + 4*1.2,0.2)--(0.6 + 4*1.2,0.2+\hw)--(0.3 + 4*1.2,0.2+\hw)--(0.3 + 4*1.2,0.2);
            \draw[thick,-](0 + 4*1.2,0)--(0 + 4*1.2,1.2);
            \draw[thick,-](0.6 + 4*1.2,0)--(0.6 + 4*1.2,1.2);
            \node at (0.3 + 0.0  + 4*1.2,0.7+\hw*0.5) {$a$};
            \node at (0.3 + 0.15 + 4*1.2,0.2+\hw*0.5) {$b$};

            \filldraw[fill=blue!60!white,-](0.0 + 5*1.2,0.7)--(0.3 + 5*1.2,0.7)--(0.3 + 5*1.2,0.7+\hw)--(0.0 + 5*1.2,0.7+\hw)--(0.0 + 5*1.2,0.7);
            \filldraw[fill=yellow!80!white,-](0.0 + 5*1.2,0.2)--(0.6 + 5*1.2,0.2)--(0.6 + 5*1.2,0.2+\hw)--(0.0 + 5*1.2,0.2+\hw)--(0.0 + 5*1.2,0.2);
            \draw[thick,-](0 + 5*1.2,0)--(0 + 5*1.2,1.2);
            \draw[thick,-](0.6 + 5*1.2,0)--(0.6 + 5*1.2,1.2);
            \node at (0.3 - 0.15 + 5*1.2,0.7+\hw*0.5) {$a$};
            \node at (0.3 + 0.0  + 5*1.2,0.2+\hw*0.5) {$b$};

            \filldraw[fill=blue!60!white,-](0.3 + 6*1.2,0.7)--(0.6 + 6*1.2,0.7)--(0.6 + 6*1.2,0.7+\hw)--(0.3 + 6*1.2,0.7+\hw)--(0.3 + 6*1.2,0.7);
            \filldraw[fill=yellow!80!white,-](0.0 + 6*1.2,0.2)--(0.6 + 6*1.2,0.2)--(0.6 + 6*1.2,0.2+\hw)--(0.0 + 6*1.2,0.2+\hw)--(0.0
            + 6*1.2,0.2);
            \draw[thick,-](0 + 6*1.2,0)--(0 + 6*1.2,1.2);
            \draw[thick,-](0.6 + 6*1.2,0)--(0.6 + 6*1.2,1.2);
            \node at (0.3 + 0.15 + 6*1.2,0.7+\hw*0.5) {$a$};
            \node at (0.3 + 0.0  + 6*1.2,0.2+\hw*0.5) {$b$};
            
        \end{scope}
    \end{tikzpicture}
  
\caption{Sparse decomposition of Ex.~\ref{ex:interval-wipoms}.}
\label{fig:sparse-example}
\end{figure}
\begin{remark}
\label{rem:criterion-sparse-dec}
In Th.~\ref{th:sparse-dec}, $\omega$-ipomsets are in $\wiipoms$ so supposed valid, thus we do not have to verify that the infinite product effectively gives a valid result. Indeed, an infinite product of alternating proper starters and terminators may yield a non-valid $\omega$-ipomset (see Fig.~\ref{fig:infinite-product-evord}).
\end{remark}

\begin{figure}[!h]
  
    \centering
    \begin{tikzpicture}[y=1cm]
        
        \def\hw{0.2}; 
        \def\shift{0.3143}; 
        \def\fs{1.3}; 
        
        \begin{scope}[shift={(0,\fs)}]

            \node at (-0.7,0) {$R =$};
            
            \node at (0,0) {%
                $\left[ \vcenter{\hbox{\!%
                \begin{tikzpicture}[x=1.2cm]
                    \node (a1) at (0,0) {$a \ibu$};
                    \node (a2) at (0,-0.25) {$a \ibu$};
                \end{tikzpicture}
              \!\!}} \right]$};

            \node at (1.2,0) {%
                $\left[ \vcenter{\hbox{\!%
                \begin{tikzpicture}[x=1.2cm]
                    \node (a1) at (0,0) {$\ibu a \nibu$};
                    \node (a2) at (0,-0.25) {$\ibu a \ibu$};
                \end{tikzpicture}
              \!\!}} \right]$};
            
            \node at (2.4,0) {%
                $\left[ \vcenter{\hbox{\!%
                \begin{tikzpicture}[x=1.2cm]
                    \node (a1) at (0,0) {$\ibu a \ibu$};
                    \node (a4) at (0,-0.25) {$\nibu a \ibu$};
                    \node (a5) at (0,-0.5) {$\nibu a \ibu$};
                \end{tikzpicture}
              \!\!}} \right]$};
            
            \node at (3.6,0) {%
                $\left[ \vcenter{\hbox{\!%
                \begin{tikzpicture}[x=1.2cm]
                    \node (a1) at (0,0) {$\ibu a \nibu$};
                    \node (a2) at (0,-0.25) {$\ibu a \ibu$};
                    \node (a3) at (0,-0.5) {$\ibu a \ibu$};
                \end{tikzpicture}
              \!\!}} \right]$};

            \node at (0.6,0) {*};
            \node at (1.7,0) {*};
            \node at (3,0) {*};
            \node at (4.2,0) {*};
            \node at (4.6,0.05) {$\cdots$};
            \node at (5,0.05) {$=$};
        
        \end{scope}

        \begin{scope}[shift={(0,-\fs)}]

            \node at (-0.7,0) {$S =$};
            
            \node at (0,0) {%
                $\left[ \vcenter{\hbox{\!%
                \begin{tikzpicture}[x=1.2cm]
                    \node (a1) at (0,0) {$a \ibu$};
                    \node (a2) at (0,-0.25) {$a \ibu$};
                \end{tikzpicture}
              \!\!}} \right]$};

            \node at (1.2,0) {%
                $\left[ \vcenter{\hbox{\!%
                \begin{tikzpicture}[x=1.2cm]
                    \node (a1) at (0,0) {$\ibu a \ibu$};
                    \node (a2) at (0,-0.25) {$\ibu a \nibu$};
                \end{tikzpicture}
              \!\!}} \right]$};
            
            \node at (2.4,0) {%
                $\left[ \vcenter{\hbox{\!%
                \begin{tikzpicture}[x=1.2cm]
                    \node (a1) at (0,0) {$\ibu a \ibu$};
                    \node (a4) at (0,-0.25) {$\nibu a \ibu$};
                    \node (a5) at (0,-0.5) {$\nibu a \ibu$};
                \end{tikzpicture}
              \!\!}} \right]$};
            
            \node at (3.6,0) {%
                $\left[ \vcenter{\hbox{\!%
                \begin{tikzpicture}[x=1.2cm]
                    \node (a1) at (0,0) {$\ibu a \ibu$};
                    \node (a2) at (0,-0.25) {$\ibu a \ibu$};
                    \node (a3) at (0,-0.5) {$\ibu a \nibu$};
                \end{tikzpicture}
              \!\!}} \right]$};
            
            \node at (0.6,0) {*};
            \node at (1.7,0) {*};
            \node at (3,0) {*};
            \node at (4.2,0) {*};
            \node at (4.6,0.05) {$\cdots$};
            \node at (5,0.05) {$=$};
        
        \end{scope}
       
        \begin{scope}[shift={(4.5,\fs)},y=1cm]
    
            \draw[thick,-](1,-0.9)--(1,1.1);

            \draw[dashed,gray,-](2 + 0, -0.9)--(2 + 0, 1.1);
            \draw[dashed,gray,-](2 + 1, -0.9)--(2 + 1, 1.1);
            \draw[dashed,gray,-](2 + 2, -0.9)--(2 + 2, 1.1);
            \draw[dashed,gray,-](2 + 3, -0.9)--(2 + 3, 1.1);
            \draw[dashed,gray,-](2 + 4, -0.9)--(2 + 4, 1.1);
            \draw[dashed,gray,-](2 + 5, -0.9)--(2 + 5, 1.1);

            \filldraw[fill=yellow!60!white,-](1.2 + 0, 0.8 - 0*\shift)--(1.8 + 1, 0.8 - 0*\shift)--(1.8 + 1, 0.8 - 0*\shift + \hw)--(1.2 + 0, 0.8 - 0*\shift + \hw)--(1.2 + 0, 0.8 - 0*\shift);              
            \node at (2, 0.8 - 0*\shift + \hw*0.5) {$a$};

            \filldraw[fill=yellow!60!white,-](1.2 + 0, 0.8 - 1*\shift)--(1.8 + 3, 0.8 - 1*\shift)--(1.8 + 3, 0.8 - 1*\shift + \hw)--(1.2 + 0, 0.8 - 1*\shift + \hw)--(1.2 + 0, 0.8 - 1*\shift);              
            \node at (3, 0.8 - 1*\shift + \hw*0.5) {$a$};

            \filldraw[fill=yellow!60!white,-](1.2 + 2, 0.8 - 2*\shift)--(1.8 + 5, 0.8 - 2*\shift)--(1.8 + 5, 0.8 - 2*\shift + \hw)--(1.2 + 2, 0.8 - 2*\shift + \hw)--(1.2 + 2, 0.8 - 2*\shift);  
            \node at (5, 0.8 - 2*\shift + \hw*0.5) {$a$};

            \filldraw[fill=yellow!60!white,-](1.2 + 2, 0.8 - 3*\shift)--(7, 0.8 - 3*\shift)--(7, 0.8 - 3*\shift + \hw)--(1.2 + 2, 0.8 - 3*\shift + \hw)--(1.2 + 2, 0.8 - 3*\shift); 
            \node at (5, 0.8 - 3*\shift + \hw*0.5) {$a$};

            \filldraw[fill=yellow!60!white,-](1.2 + 4, 0.8 - 4*\shift)--(7, 0.8 - 4*\shift)--(7, 0.8 - 4*\shift + \hw)--(1.2 + 4, 0.8 - 4*\shift + \hw)--(1.2 + 4, 0.8 - 4*\shift); 
            \node at (6, 0.8 - 4*\shift + \hw*0.5) {$a$};

            \filldraw[fill=yellow!60!white,-](1.2 + 4, 0.8 - 5*\shift)--(7, 0.8 - 5*\shift)--(7, 0.8 - 5*\shift + \hw)--(1.2 + 4, 0.8 - 5*\shift + \hw)--(1.2 + 4, 0.8 - 5*\shift); 
            \node at (6, 0.8 - 5*\shift + \hw*0.5) {$a$};

            \node at (7.5,0) {$\cdots$};
    
        \end{scope}

        \begin{scope}[shift={(4.5,-\fs)},y=1cm]

            \draw[thick,-](1,-0.9)--(1,1.1);

            \draw[dashed,gray,-](2 + 0, -0.9)--(2 + 0, 1.1);
            \draw[dashed,gray,-](2 + 1, -0.9)--(2 + 1, 1.1);
            \draw[dashed,gray,-](2 + 2, -0.9)--(2 + 2, 1.1);
            \draw[dashed,gray,-](2 + 3, -0.9)--(2 + 3, 1.1);
            \draw[dashed,gray,-](2 + 4, -0.9)--(2 + 4, 1.1);
            \draw[dashed,gray,-](2 + 5, -0.9)--(2 + 5, 1.1);

            \filldraw[fill=blue!60!white,-](1.2 + 0, 0.8 - 0*\shift)--(7, 0.8 - 0*\shift)--(7, 0.8 - 0*\shift + \hw)--(1.2 + 0, 0.8 - 0*\shift + \hw)--(1.2 + 0, 0.8 - 0*\shift);    
            \node at (4, 0.8 - 0*\shift + \hw*0.5) {$a$};

           \filldraw[fill=yellow!60!white,-](1.2 + 0, 0.8 - 1*\shift)--(1.8 + 1, 0.8 - 1*\shift)--(1.8 + 1, 0.8 - 1*\shift + \hw)--(1.2 + 0, 0.8 - 1*\shift + \hw)--(1.2 + 0, 0.8 - 1*\shift);              
            \node at (2, 0.8 - 1*\shift + \hw*0.5) {$a$};

            \filldraw[fill=blue!60!white,-](1.2 + 2, 0.8 - 2*\shift)--(7, 0.8 - 2*\shift)--(7, 0.8 - 2*\shift + \hw)--(1.2 + 2, 0.8 - 2*\shift + \hw)--(1.2 + 2, 0.8 - 2*\shift);  
            \node at (5, 0.8 - 2*\shift + \hw*0.5) {$a$};

           \filldraw[fill=yellow!60!white,-](1.2 + 2, 0.8 - 3*\shift)--(1.8 + 3, 0.8 - 3*\shift)--(1.8 + 3, 0.8 - 3*\shift + \hw)--(1.2 + 2, 0.8 - 3*\shift + \hw)--(1.2 + 2, 0.8 - 3*\shift);              
            \node at (4, 0.8 - 3*\shift + \hw*0.5) {$a$};

            \filldraw[fill=blue!60!white,-](1.2 + 4, 0.8 - 4*\shift)--(7, 0.8 - 4*\shift)--(7, 0.8 - 4*\shift + \hw)--(1.2 + 4, 0.8 - 4*\shift + \hw)--(1.2 + 4, 0.8 - 4*\shift); 
            \node at (6, 0.8 - 4*\shift + \hw*0.5) {$a$};

            \filldraw[fill=yellow!60!white,-](1.2 + 4, 0.8 - 5*\shift)--(1.8 + 5, 0.8 - 5*\shift)--(1.8 + 5, 0.8 - 5*\shift + \hw)--(1.2 + 4, 0.8 - 5*\shift + \hw)--(1.2 + 4, 0.8 - 5*\shift); 
            \node at (6, 0.8 - 5*\shift + \hw*0.5) {$a$};

            \node at (7.5,0) {$\cdots$};          
    
        \end{scope}

  \end{tikzpicture}
  
\caption{Two sparse decompositions with different event orders, such that only $R$ is valid.}
\label{fig:infinite-product-evord}
\end{figure}
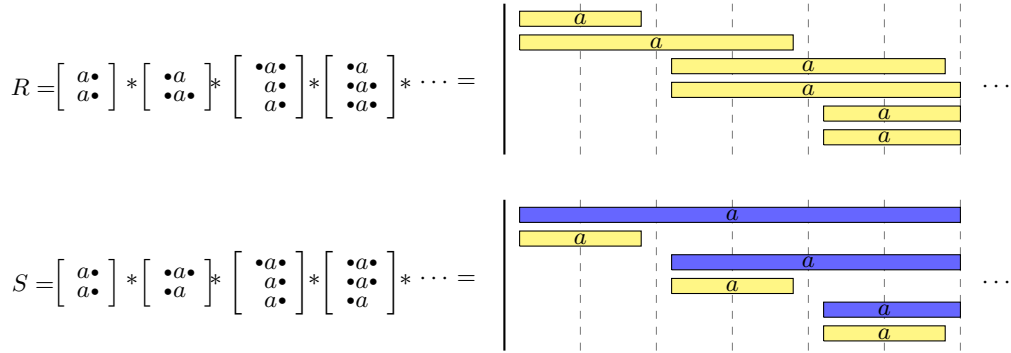

Defining such a decomposition gives us a characterization of $\omega$-iipomsets in terms of their prefixes.
We say that $A \in \ipoms$ is a (finite) \emph{prefix} of $P \in \infipoms$ if there is $Q \in \infipoms$ such that $P = AQ$, and then write $A \prf P$. We call $\pref{P}$ the set of all prefixes of $P$.

\begin{remark}
    Unlike classical $\omega$-words theory, two prefixes of $P$ may be incomparable, for example $a\ibu \prf \loset{a\\b}$ and $b\ibu \prf \loset{a\\b}$ but $a\ibu \not\prf b\ibu$ and $b\ibu \not\prf a\ibu$. But we still have the property that for $A \in \pref{P}$ and $x \in A$, if $y <_P x$ then $y \in A$. Note that neither $a\ibu$ nor $b\ibu$ could be the first step of the sparse decomposition of $(a\parallel b)^\omega$. Its only possible first step is $\loset{a\ibu \\b \ibu }$.
\end{remark}

The following lemma is trivial in classical $\omega$-word theory, but here relies crucially on the fact that we deal with valid $\omega$-ipomsets.

\begin{lemma}
\label{lem:pref-exist}
    If $(A_i)_{i \in I}$ is an infinite set of prefixes of an $\omega$-ipomset $P$, and $R$ is a finite subset of $P$ (seen as a set), then there is $i\in I$ such that $R \subset A_i$.
\end{lemma}
\begin{proof}
	Let $(A_i)_{i \in I} \in \pref{P}^I$ (with $I$ infinite), and $x \in R$. We want to prove that the set $\set{A_i \midbar x \in A_i} $ is infinite. Let $x \in P$, and $P^m$ be (the antichain of) the $<$-minimal elements of $\set{x} \cup \set{p \in P \midbar x \not< p \text{ and } p \not< x}$, and $Q := P^m \cup \set{p \in P \midbar p < x}$. As $P$ is valid, $Q$ is finite, and so there is a finite number of ipomsets included in $Q$. Thus there is an infinite number of $A_i \sth A_i \cap (P - Q) \neq \emptyset$. In such $A_i$, let $y_i \in A_i \cap (P - Q)$. There is an element $z_i \in P^m \sth z_i < y_i$: either $x < y_i$ (and so $z_i = x$), or $y_i \in \set{p \in P \midbar x \not< p \text{ and } p \not< x}$ (and so there is $z_i \in P^m \subset Q$ such that $z_i \leq y_i$, and $z_i \neq y_i$ as $z_i \in Q$ and $y_i \not\in Q$). As $A_i$ is a prefix of $P$, $y_i \in A_i$ implies that $z_i$ is necessarily in $A_i - T_{A_i}$, and so $x \in A_i$ (otherwise $x \in B_i - S_{B_i}$ with $P = A_i*B_i$ and so $z_i < x$, which is impossible because $P^m$ is an antichain). So there are infinitely many $A_i$ such that $x \in A_i$.

    Then, we take $(A_i)_{i \in I'}$, an infinite set of prefixes of $P$ that all contain $x$, and we iteratively use the previous argument to get all elements of the (finite) set $R$. At the end we have $(A_i)_{i \in J}$, an infinite set of prefixes of $P$ that all contain $R$, and we conclude the proof by taking anyone of them.
\end{proof}
\begin{remark}
\label{rem:pref-counter-example}
    Note that $P$ is not necessarily  interval in Lem.~\ref{lem:pref-exist}. 
    It only suffices to have an infinite number of prefixes. However being an interval $\omega$-ipomset ensures the existence of an infinite set of prefixes, see below.
    For example, $P = a^\omega \parallel b^\omega \not\in \wiipoms$ has a finite number of prefixes: $\pref{P} = \set{a \ibu,b\ibu,\loset{a\ibu \\ b\ibu}}$, while $Q = (aa \parallel bb).c^\omega \not\in \wiipoms$ has an infinite number of prefixes: $ \pref{Q}\supseteq \set{(aa \parallel bb).c^n \midbar n \in \Nat}$.
\end{remark}
From Lem.~\ref{lem:pref-exist} follows a sequential characterization of interval $\omega$-ipomsets:
\begin{proposition}
\label{prop:seq-wiipoms}
    Let $P$ be an $\omega$-ipomset, the following are equivalent:
    \begin{enumerate}
        \item 
            $P$ is an interval $\omega$-ipomset;
        \item
            $\pref{P}$ is infinite and for all $A \in \pref{P}$, $A$ is an interval ipomset;
        \item
            an infinite number of prefixes of $P$ are interval ipomsets.
    \end{enumerate}
\end{proposition}

In addition, from the last point of Prop.~\ref{prop:seq-wiipoms} along with the preservation of the interval property by gluing composition (see~\cite[Lem.~41]{DBLP:journals/mscs/FahrenbergJSZ21}) we have the following:

\begin{corollary}
\label{cor:inf-product-interval}
    If $P = P_0 * P_1 * \cdots$ is a valid $\omega$-ipomset and $\forall i \in \Nat,\, P_i \in \iipoms$, then $P \in \wiipoms$.
\end{corollary}

\section{Higher-dimensional automata}

In this section, we recall  \emph{higher-dimensional automata} (\emph{HDAs}) over iipomsets~\cite{amrane.23.ictac, DBLP:journals/mscs/FahrenbergJSZ21, DBLP:journals/corr/abs-2210-08298} and introduce $\omega$-HDAs.
Recall that $\conc$ denotes the set of conclists.

\subsection{HDAs over finite ipomsets}
\label{sec:finite-HDA}

	A \emph{precubical set} is a structure $(X,\ev,\Delta)$ with:
	\begin{itemize}
		\item
		a set of cells $X$;
		\item
		a function $\ev : X \maps \conc$ which assigns to a cell its list of active events (for $U \in \conc$ we write $X[U] = \set{q \in X \midbar \ev(x) = U}$); 
		\item
		face maps $\Delta = \set{\delta_{A,U}^0,\,  \delta_{A,U}^1\midbar U \in \conc,\, A \subset U}$ such that $\delta_{A,U}^0,\delta_{A,U}^1 : X[U] \maps X[U-A]$;
		\item 
		for $A,B \subset U \sth A \cap B = \emptyset$ and $\mu, \nu \in \set{0,1}$, we have the following \emph{precubical identities}: $\delta_{A,U-B}^\mu \circ \delta_{B,U}^\nu = \delta_{B,U-A}^\nu \circ \delta_{A,U}^\mu$.
	\end{itemize}
We usually refer to a precubical set by its set of cells $X$, and for face maps we often omit the second subscript $U$. The \emph{dimension} of a cell $q \in X$ is the size $|\ev(q)|$.
An \emph{upper} face map $\delta_A^1$ terminates the events in $A$, whereas a \emph{lower} face map $\delta_A^0$ ``unstarts'' these events, that is, it maps to a cell where the events of $A$ are not yet started. 

A \emph{higher-dimensional automaton} is a tuple $(X, \bot_X, \top_X)$ where $X$ is a \emph{finite} precubical set, and $\bot_X, \top_X\subset X$ are the sets of \emph{starting} and \emph{accepting} cells. 
We may omit the subscripts $_X$ if the context is clear, and refer to an HDA only by its precubical set.

\begin{example}
	\label{ex:hda}
    Fig.~\ref{fig:hda-first-example} shows a two-dimensional HDA as a combinatorial object (left) and in a geometric realisation (right).
    The arrows between the cells on the left representation correspond to the face maps connecting them.   
    It consists of nine cells: the corner cells $X_0 = \{x,y,v,w\}$ in which no event is active (for all $z \in X_0$, $\ev(z) = \emptyset$), the transition cells $X_1 = \{g,h,f,e\}$ in which one event is active ($\ev(f) = \ev(e) = a$ and $\ev(g) = \ev(h) = b$), and the square cell $q$ where $\ev(q) = \loset{a\\b}$.
    By convention, for $a \evord b$, we represent the event $a$ horizontally and the event $b$ vertically.
\end{example}

\begin{figure}[h!]
	\centering
	\begin{tikzpicture}[x=.9cm, y=.9cm, scale=0.9, every node/.style={transform shape}]
		\node[circle,draw=black,fill=blue!30,inner sep=0pt,minimum size=15pt]
		(aa) at (0,0) {$\vphantom{hy}v$};
		\node[circle,draw=black,fill=blue!30,inner sep=0pt,minimum size=15pt]
		(ac) at (0,4) {$\vphantom{hy}x$};
		\node[circle,draw=black,fill=blue!30,inner sep=0pt,minimum size=15pt]
		(ca) at (4,0) {$\vphantom{hy}w$};
		\node[circle,draw=black,fill=blue!30,inner sep=0pt,minimum size=15pt]
		(cc) at (4,4) {$\vphantom{hy}y$};
		\node[circle,draw=black,fill=red!30,inner sep=0pt,minimum size=15pt]
		(ba) at (2,0) {$\vphantom{hy}e$};
		\node[circle,draw=black,fill=red!30,inner sep=0pt,minimum size=15pt]
		(bc) at (2,4) {$\vphantom{hy}f$};
		\node[circle,draw=black,fill=green!30,inner sep=0pt,minimum size=15pt]
		(ab) at (0,2) {$\vphantom{hy}g$};
		\node[circle,draw=black,fill=green!30,inner sep=0pt,minimum size=15pt]
		(cb) at (4,2) {$\vphantom{hy}h$};
		\node[circle,draw=black,fill=black!20,inner sep=0pt,minimum size=15pt]
		(bb) at (2,2) {$\vphantom{hy}q$};
		\node[right] at (5,4) {$X[\emptyset]=\{v,w,x,y\}$};
		\node[right] at (5,3.2) {$X[a]=\{e,f\}$};
		\node[right] at (5,2.4) {$X[b]=\{g,h\}$};
		\node[right] at (5,1.6) {$X[\loset{a\\b}]=\{q\}$};
		\path (ba) edge node[above] {$\delta^0_a$} (aa);
		\path (ba) edge node[above] {$\delta^1_a$} (ca);
		\path (bb) edge node[above] {$\delta^0_a$} (ab);
		\path (bb) edge node[above] {$\delta^1_a$} (cb);
		\path (bc) edge node[above] {$\delta^0_a$} (ac);
		\path (bc) edge node[above] {$\delta^1_a$} (cc);
		\path (ab) edge node[left] {$\delta^0_b$} (aa);
		\path (ab) edge node[left] {$\delta^1_b$} (ac);
		\path (bb) edge node[left] {$\delta^0_b$} (ba);
		\path (bb) edge node[left] {$\delta^1_b$} (bc);
		\path (cb) edge node[left] {$\delta^0_b$} (ca);
		\path (cb) edge node[left] {$\delta^1_b$} (cc);
		\path (bb) edge node[above left] {$\delta^1_{ab}\!\!$} (cc);
		\path (bb) edge node[above left] {$\delta^0_{ab}\!\!$} (aa);
		\node[below left] at (aa) {$\bot\;$};
		\node[below left] at (ab) {$\bot\;$};
		\node[above right] at (cb) {$\;\top$};
		\node[above right] at (cc) {$\;\top$};
		\node[above right] at (ab) {$\;\top$};
		\node[right] at (5,0.8) {$\bot_X=\{v, g\}$};
		\node[right] at (5,0) {$\top_X=\{h, y, g\}$};
		\begin{scope}[shift={(8.8,.4)}, x=1.3cm, y=1.3cm]
			\filldraw[color=black!20] (0,0)--(2,0)--(2,2)--(0,2)--(0,0);			

			\path[red!50!black,line width=1] (0,0) edge node[below, black] {$\vphantom{b}a$} (1.95,0);
			\path[red!50!black,line width=1] (0,2) edge (1.95,2);
			\path[green!50!black,line width=1] (0,0) edge node[pos=.6, left, black] {$\vphantom{bg}b$} (0,1.95);
			\path[green!50!black,line width=1] (2,0) edge (2,1.95);
			\node[left] at (0,0) {$\bot$};
			\node[left] at (0,0.7) {$\bot$};
			\node[right] at (0,0.7) {$\top$};
			\node[right] at (2,2) {$\top$};
			\node[right] at (2,1) {$\top$};
			
			\node[blue!70,centered] at (0,-0.2) {$v$};
			\node[centered, red!50!black] at (1,0.15) {$e$};
			\node[blue!70,centered] at (2,-0.2) {$w$};
			\node[centered,blue!70] at (2,2.2) {$y$};
			\node[centered,blue!70] at (0,2.2) {$x$};
			\node[centered, green!50!black] at (0.2,1.1) {$\vphantom{bg}g$};
			\node[centered, green!50!black] at (1.8,1.1) {$\vphantom{bg}h$};
			\node[centered] at (1,1) {$q$};
			\node[centered, red!50!black] at (1,1.75) {$f$};

            \filldraw (0,0) circle (0.05);
			\filldraw (2,0) circle (0.05);
			\filldraw (0,2) circle (0.05);
			\filldraw (2,2) circle (0.05);
            
		\end{scope}
	\end{tikzpicture}
	\caption{A two-dimensional HDA $X$ on $\Sigma=\{a, b\}$, see Ex.~\ref{ex:hda}.}
	\label{fig:hda-first-example}
\end{figure}
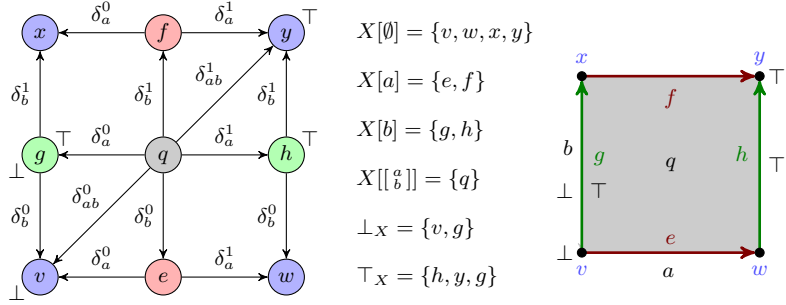

A \emph{track} $\alpha$ in $X$ is a sequence $(q_0,\phi_1,q_1, \dots,\phi_n,q_n)$ where $q_i \in X$ are cells and $\phi_i$ denote face map types. That is, for all $i \leq n$, $(q_i,\phi_i,q_{i\+1})$ is:
\begin{itemize}
    \item
    either an \emph{upstep}: $(\delta_A^0(q_{i\+1}),\arrO{A},q_{i\+1})$ with $A \subset \ev(q_{i\+1})$,
    \item
    or a \emph{downstep}: $(q_i,\arrI{B},\delta_B^1(q_i))$ with $B \subset \ev(q_i)$.
\end{itemize}
The \emph{source} of  $\alpha$ is  $\src(\alpha) = q_0$ and its \emph{target} is $\tgt(\alpha) = q_n$. A track is \emph{accepting} if $\src(\alpha) \in \bot_X$ and $\tgt(\alpha) \in \top_X$.
Note that tracks are concatenations (denoted using $*$) of upsteps and downsteps. 
The language of an HDA is defined in terms of the event ipomsets of its accepting tracks.
Formally, the \emph{event ipomset} of a track $\alpha$ (written $\ev(\alpha)$) is defined by:
\begin{itemize}
    \item
        $\ev((q)) = \id_{\ev(q)}$ (with $q \in X$ and $(q)$ a single-cell track)
    \item
        $\ev((q,\arrO{A},p)) = \starter{\ev(p)}{A}$ (the starter which starts events $A$)
    \item
        $\ev((q,\arrI{B},p)) = \terminator{\ev(q)}{B}$ (the terminator which terminates events $B$)
    \item 
        $\ev(\alpha_1 * \dots * \alpha_n) = \ev(\alpha_1) * \dots * \ev(\alpha_n)$
\end{itemize}
We say that two tracks $\alpha$ and $\beta$ are equivalent, denoted $\alpha \simeq \beta$, with $\simeq$ the relation generated by the three cases $(x \arrO{A} y \arrO{B} z)\simeq (x \arrO{A\cup B} z)$, $(x \arrI{A} y \arrI{B} z) \simeq (x \arrI{A\cup B} z)$, and $\gamma \alpha \delta\simeq \gamma \beta \delta$ whenever $\alpha\simeq \beta$. Basically, two tracks are equivalent if their ``contractions'' into sparse tracks (i.e.\ alternating upsteps and downsteps) and event ipomsets are equal.

\begin{example}
\label{ex:tracks}
    The HDA $X$ of Fig.~\ref{fig:hda-first-example} admits several accepting tracks with target $h$, for example $v\arrO{ab} q\arrI{a} h$.
    This is a sparse track and equivalent to the non-sparse tracks $v\arrO{a} e\arrO{b} q\arrI{a} h$ and $v\arrO{b} g\arrO{a} q\arrI{a} h$.
    Their event ipomset is $\loset{a\nibu\\b\ibu}$.
    The track $v\arrO{a} e \arrI{a} w \arrO{b} h$ is also accepting with event ipomset $ab\ibu \subsu \loset{a\nibu\\b\ibu}$.
    In addition, since $g$ is both a start and accept cell, we have also $g$ and $v\arrO{b} g$ as accepting tracks, with event ipomsets $\ibu b \ibu$ and $b \ibu$, respectively.
    Note that tracks move forward along upper face maps but backward along lower ones.
    $X$ accept tracks whose ipomsets are in $\set{b\ibu,\ibu b\ibu,ab\ibu,\loset{a\nibu\\b\ibu}, \loset{a\\\ibu b\ibu},\ibu b a,\loset{\nibu a\\\ibu b},ab,ba,\loset{a\\b}}$.
\end{example}

The \emph{language} of an HDA $X$ is $L(X) = \set{\ev(\alpha) \midbar \src(\alpha) \in \bot_X \text{ and } \tgt(\alpha) \in \top_X}$. The following result about languages of HDAs is a fundamental property:

\begin{proposition}[{\cite[Prop.~10]{DBLP:conf/concur/FahrenbergJSZ22}}]
	\label{prop:HDA-down-closed}
	Languages of HDAs over finite ipomsets are closed under subsumption (or down-closed).
\end{proposition}

For $L \subset \ipoms$, we write $L \down = \set{P \in \iipoms \midbar \exists Q \in L,\, P \subsu Q}$ for the subsumption closure of~$L$.
For example, the language of Ex.~\ref{ex:hda} is $L(X)=\{b\ibu,\ibu b\ibu,\loset{a\nibu\\b\ibu}, \loset{a\\\ibu b\ibu},\loset{a\\b},\loset{\nibu a\\\ibu b}\}\down$.
Let us extend the notation to $R \subset 2^\ipoms$ with $R\down = \set{L\down \midbar L \in R}$.

Previous work \cite{DBLP:conf/concur/FahrenbergJSZ22} defines \emph{down-closed rational operations} $\cup$, $*\sd$, $\|\sd$ and (Kleene plus) $^{+\sd}$ for languages of finite iipomsets, as follows:
\begin{itemize}
\item $L *\sd M = \{P*Q\mid P\in L,\; Q\in M,\; T_P=S_Q\}\down$,
\item $L\parallel\sd M = \{P\parallel Q\mid P\in L,\; Q\in M\}\down$,
\item $L^{+\sd} = \bigcup\nolimits_{n\ge 1} L^n$, with $L^1=L$ and $L^{n+1}= L*\sd L^n$.
\end{itemize}
The class $\Drat$ of \emph{down-closed rational languages} is then defined to be the smallest class that contains $\emptyset$, $\set{\eps}$, $\set{a}$, $\set{\ibu a}$,$\set{a\ibu}$, $\set{\ibu a \ibu}$ (for $a\in \Sigma$), and is closed under the operations above. 
Note that the explicit downclosure $\down$ of the operations above ensures that the built languages contain only interval ipomsets.\footnote{For $L = \set{ac}$ and $M = \set{\ibu b d \ibu}$, $(ac \parallel\ibu bd \ibu) \not\in L \parallel\sd M$ but, for example, the ipomsets of Fig~\ref{fig:interval-ipomset} are in $L \parallel\sd M$.
}
We will below introduce rational operations which \emph{do not} apply down-closure;
to distinguish them from the ones above we have added a subscript $\sd$.

\begin{theorem}[\cite{DBLP:conf/concur/FahrenbergJSZ22}]
  \label{th:kleene}
  A language in $2^\iipoms$ can be recognized by an HDA iff it is in $\Drat$.
\end{theorem}

\subsection{$\omega$-higher-dimensional automata}
\label{sec:def-wHDA}

An \emph{$\omega$-HDA} is an HDA whose tracks are infinite. 
Formally, an \emph{$\omega$-track} is an infinite sequence $\alpha = (q_0,\phi_1,q_1, \dots)$ where $q_i \in X$ and each $(q_i,\phi_i,q_{i\+1})$ is an upstep or a downstep. 
The event ipomset of an $\omega$-track is defined similarly using the infinite gluing. 
From this we may directly introduce the classical acceptance conditions of $\omega$-automata theory:

Let $\Inf(\alpha) = \set{q \in X \mid |\{i \midbar q_i = q\}| = + \infty}$ be the set of cells seen infinitely often by $\alpha$.
\label{def:Büchi-w-HDA}
A \emph{Büchi} $\omega$-HDA is an HDA $(X, \bot_X, \top_X)$ where $ \bot_X, \top_X\subset X$ are the sets of \emph{starting} and \emph{accepting} cells, and where an $\omega$-track $\alpha$ is accepting if $\src(\alpha) \in \bot_X$ and $\Inf(\alpha) \cap \top_X \neq \emptyset$.
Similarly,  a \emph{Muller} $\omega$-HDA is an HDA $(X, \bot_X, F_X)$ where $ \bot_X \subset X$ is the set of \emph{starting} cells and $F_X \subset 2^X$ is the set of \emph{accepting sets} of cells, and where an $\omega$-track $\alpha$ is accepting if $\src(\alpha) \in \bot_X$ and $\Inf(\alpha) \in F_X$.

\begin{example}
\label{ex:buchi}
	Büchi $\omega$-HDAs $X_1,X_2$ and $X_3$ are defined in Fig.~\ref{fig:buchi-example}. 
	Cells of the same color are identified in each $\omega$-HDA, and so are their respective faces (as required by the definition of precubical sets).
    Observe that $X_1$ and $X_2$ form a torus and $X_3$ a cylinder.
   Their languages will be specified later in the paper (see Ex.~\ref{ex:wrat-languages} and Ex.~\ref{ex:buchi-wform}). 
\end{example}

\begin{example}
\label{ex:muller}
    Let $X_4$ be the torus HDA of Fig.~\ref{fig:buchi-example} and assume $\bot_{X_4} = \{q_0\}$ and $F_{X_4} = \{\{q_0,q_{ab}\}\}$.
    Then $(X_4,\bot_{X_4},F_{X_4})$ defines a Muller $\omega$-HDA.
    The $\omega$-track $\alpha = (q_0,\arrO{a},q_a,\arrI{a},q_0,\arrO{ab},q_{ab},\arrI{ab},q_0,\arrO{ab},q_{ab},\arrI{ab},q_0,\cdots)$ of event $\ev(\alpha)= a \cdot \loset{a\\b}^\omega$ is accepting as $\Inf(\alpha) = \set{q_0,q_{ab}} \in F_{X_4}$.
\end{example}

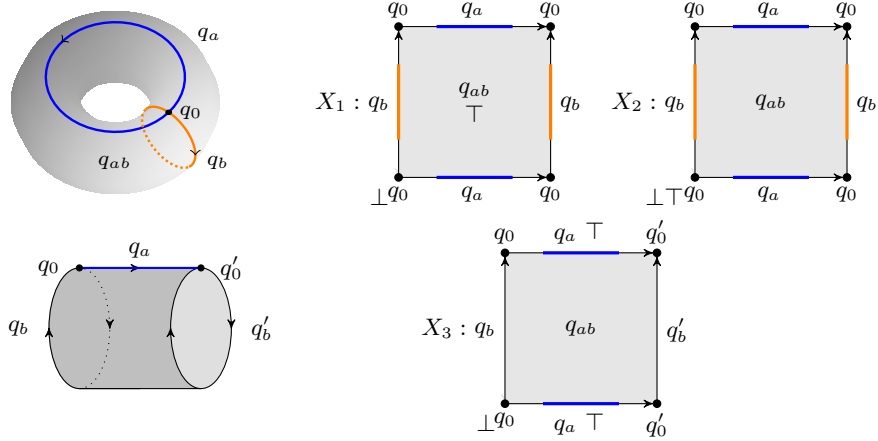
\begin{figure}[!h]

    \centering
    \begin{tikzpicture}[scale=1]

       \begin{scope}[shift={(0,1)}]
			\node at (0,0) {\includegraphics[width=3cm]{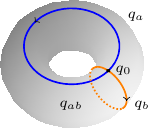}};
		\end{scope}	
        \begin{scope}[shift={(3.6,0)}]

            \node at (-0.8,1) {$X_1:$};
            
            \filldraw[color=black!10!white] (0,0)--(2,0)--(2,2)--(0,2)--(0,0);

            \filldraw (0,0) circle (0.05);
            \filldraw (0,2) circle (0.05);
            \filldraw (2,0) circle (0.05);
            \filldraw (2,2) circle (0.05);
            
            \path (0,0) edge node[below] {$q_a$} (1.95,0);
			\path (0,2) edge node[above] {$q_a$} (1.95,2);
			\path (0,0) edge node[left] {$q_b$} (0,1.95);
			\path (2,0) edge node[right] {$q_b$} (2,1.95);

            \node[below] at (0,0) {$q_0$};
            \node[below] at (2,0) {$q_0$};
            \node[above] at (0,2) {$q_0$};
            \node[above] at (2,2) {$q_0$};

            \path (0.5,0) edge[color=blue,-, very thick] (1.5,0);
            \path (0.5,2) edge[color=blue,-, very thick] (1.5,2);
            \path (0,0.5) edge[color=orange,-, very thick] (0,1.5);
            \path (2,0.5) edge[color=orange,-, very thick] (2,1.5);
            
            \node at (-0.25,-0.25) {$\bot$};
            \node at (1,0.85) {$\top$};
            \node at (1,1.15) {$q_{ab}$};

        \end{scope}

        \begin{scope}[shift={(7.5,0)}]

            \node at (-0.8,1) {$X_2:$};

            \filldraw[color=black!10!white] (0,0)--(2,0)--(2,2)--(0,2)--(0,0);
            
            \filldraw (0,0) circle (0.05);
            \filldraw (0,2) circle (0.05);
            \filldraw (2,0) circle (0.05);
            \filldraw (2,2) circle (0.05);
            
            \path (0,0) edge node[below] {$q_a$} (1.95,0);
			\path (0,2) edge node[above] {$q_a$} (1.95,2);
			\path (0,0) edge node[left] {$q_b$} (0,1.95);
			\path (2,0) edge node[right] {$q_b$} (2,1.95);

            \node[below] at (0,0) {$q_0$};
            \node[below] at (2,0) {$q_0$};
            \node[above] at (0,2) {$q_0$};
            \node[above] at (2,2) {$q_0$};
            
            \node at (1,1) {$q_{ab}$};
            \node at (-0.4,-0.25) {$\bot \top$};
            
            \path (0.5,0) edge[color=blue,-, very thick] (1.5,0);
            \path (0.5,2) edge[color=blue,-, very thick] (1.5,2);
            \path (0,0.5) edge[color=orange,-, very thick] (0,1.5);
            \path (2,0.5) edge[color=orange,-, very thick] (2,1.5);
        
        \end{scope}
        \begin{scope}[shift={(-0.6,-2)}, scale=.8, -]
			\filldraw[color=lightgray] (0,-1) rectangle (2,1);
			\filldraw[color=lightgray] (0,0) circle [x radius=.5,y radius=1];
			\draw (0,1) arc [x radius=.5,y radius=1,start angle=90,end angle=270];
			\draw[dotted] (0,-1) arc [x radius=.5,y radius=1,start angle=-90,end angle=90];
			\draw[fill=lightgray!50!white] (2,0) circle [x radius=.5,y radius=1];
			\draw[->] (-.5,-.01) -- (-.5,.01);
			\draw[->] (.5,.01) -- (.5,-.01);
			\draw[->] (1.5,-.01) -- (1.5,.01);
			\draw[->] (2.5,.01) -- (2.5,-.01);
			\draw (0,-1) edge[-] (1,-1) -- (2,-1);
			\draw (0,1) edge[->] (1,1) -- (2,1);
			\draw[blue, thick] (0,1) -- node[above, black] {$q_a$} (2,1);
			\node at (-1,0) {$q_b$};
			\node at (3,0) {$q_b'$};
                        \filldraw (0,1) circle (0.05);
                        \node at (-.5,1) {$q_0$};
                        \filldraw (2,1) circle (0.05);
                        \node at (2.5,1) {$q_0'$};
		\end{scope}

        \begin{scope}[shift={(5,-3)}]

            \node at (-0.8,1) {$X_3:$};
        
            \filldraw[color=black!10!white] (0,0)--(2,0)--(2,2)--(0,2)--(0,0);
            
            \filldraw (0,0) circle (0.05);
            \filldraw (0,2) circle (0.05);
            \filldraw (2,0) circle (0.05);
            \filldraw (2,2) circle (0.05);
            
            \path (0,0) edge node[below] {$q_a \ \top$} (1.95,0);
            \path (0,2) edge node[above] {$q_a \ \top$} (1.95,2);
            \path (0,0) edge node[left] {$q_b$} (0,1.95);
            \path (2,0) edge node[right] {$q'_b$} (2,1.95);

            \node[below] at (0,0) {$q_0$};
            \node[below] at (2,0) {$q_0'$};
            \node[above] at (0,2) {$q_0$};
            \node[above] at (2,2) {$q_0'$};

            \node at (1,1) {$q_{ab}$};
            \node at (-0.25,-0.25) {$\bot$};
            
            \path (0.5,0) edge[color=blue,-, very thick] (1.5,0);
            \path (0.5,2) edge[color=blue,-, very thick] (1.5,2);
        
        \end{scope}

    \end{tikzpicture}
    
\caption{Büchi $\omega$-HDAs (left) and unfolded views (right). See Ex.~\ref{ex:buchi}.}
\label{fig:buchi-example}
\end{figure}

The \emph{language} of a Büchi $\omega$-HDA $X$ is $L(X) = \set{ \ev(\alpha) \midbar \src(\alpha) \in \bot_X \text{ and } \Inf(\alpha) \cap \top_X \neq \emptyset}$ and we call \emph{Büchi $\omega$-regular} such a language. Similarly, for a Muller $\omega$-HDA $Y$ we have ${L(Y) = \set{ \ev(\alpha) \midbar \src(\alpha) \in \bot_Y \text{ and } \Inf(\alpha) \in F_Y}}$ and we call \emph{Muller $\omega$-regular} such a language.
As in classical $\omega$-theory, a Büchi $\omega$-HDA $X_B$ can be seen as a Muller $\omega$-HDA $X_M$, where the acceptance condition is modified to $F_{X_M} = \set{ A \subset X \midbar A \cap \top_X \neq \emptyset}$. Hence Büchi $\omega$-regular languages are Muller $\omega$-regular.
However, unlike HDAs over finite ipomsets we have:
\begin{proposition}
\label{prop:wHDA-not-down-closed}
    $\omega$-regular languages may \emph{not} be  down-closed.
\end{proposition}

\begin{proof}
  In the Büchi $\omega$-HDA $X_1$ of Ex.~\ref{ex:buchi}, we have  $\loset{a \\ b}^\omega = \ev(q_0,\arrO{ab},q_{ab},\arrI{ab},q_0,...) \in L(X_1)$ (as it goes infinitely many times through $q_{ab} \in \top_{X_1}$). But $(ab)^\omega \subsu \loset{a \\ b}^\omega$, and the only $\omega$-track whose event is this $\omega$-ipomset is $(q_0,\arrO{a},q_a,\arrI{a},q_0,\arrO{b},q_b,\arrI{b},q_0,...)$ which is not accepting. Thus $(ab)^\omega \not\in L(X_1)$. The same can be done for the Muller $\omega$-HDA $X_4$ of Ex.~\ref{ex:muller}.
\end{proof}

\section{$\omega$-rational languages}
\label{sec:w-rat}

We want to extend the notion of rationality to languages of  interval $\omega$-ipomsets and capture the expressiveness of $\omega$-HDAs.
Thus, we introduce a non-nested $\omega$-iteration over languages of finite interval ipomsets.
On the other hand, the rational operations defined in~\cite{DBLP:conf/concur/FahrenbergJSZ22} take subsumption into account. 
As languages of $\omega$-HDA are not down-closed, we cannot just extend these to $\omega$-ipomsets,
but need to redefine them without subsumption closure.

We define the $\omega$-\emph{rational operations} over $2^\infiipoms$ by:
\begin{itemize}
    \item
        For $(L,M) \subset \infiipoms \times \infiipoms$: $L + M = L \cup M$
    \item
        For $(L,M) \subset \iipoms \times \infiipoms$: $L * M = LM = \set{P * Q \midbar P \in L,\, Q \in M,\, T_P \cong S_Q}$
    \item
         For $L \subset \iipoms$: $L^+ = \bigcup_{n \geq 1}L^n$ (with $L^1 = L$ and $L^{n+1} = L * L^n$)
    \item
         For $L \subset \iipoms$: $L^\omega =  \set{P_1 * P_2 * \dots \midbar \forall i,\, P_i \in L, T_{P_i} \cong S_{P_{i\+1}} \text{ and } |\set{i : P_i \not\in \Id}| = + \infty }$
\end{itemize}
As in the finite case, the Kleene iteration is the non-empty $^+$ instead of $^*$ because $L^0 = \Id$ is not regular since it has infinite width.
In addition, the operation $\parallel$ defined over languages of finite interval ipomsets (and over finite  \cite{gischer1988equational} and $\omega$-series-parallel pomsets \cite{Kuske00}) is not used in our case. 
The reason is that it would typically not give interval $\omega$-ipomsets $(a^\omega \parallel b^\omega$ for example).
In \cite[Prop.~16]{DBLP:conf/concur/FahrenbergJSZ22}, it is shown that $\Drat$ may also be obtained from $\emptyset$ and $\Xi$ (the set of starters and terminators over $\Sigma$) using down-closed rational operations \emph{without} $\|$.
We define $\omega$-rational languages $\infrat$ as the smallest set such that:
\begin{itemize}
    \item 
        $\emptyset \in \infrat$ and for $U \in \Xi$, $\set{U} \in \infrat$;
    \item
        if $L,M \in \infrat$, then $L+M \in \infrat$;
    \item
        if $L,M \in \infrat$ and $L \subset \ipoms$, then $L*M \in \infrat$;
    \item
        if $L \in \infrat$ and $L \subset \ipoms$, then $L^+ \in \infrat$ and $L^\omega \in \infrat$.
\end{itemize}
We will often use $U$ instead of $\set{U}$ for ease of reading. We define the \emph{width of a language} $L$ by $\ms{wd}(L) = \sup{ \ms{wd}(P) \midbar P \in L}$. 
As they are built inductively from $\Xi$, using  finitely many $\omega$-rational operations, all $\omega$-rational languages have  finite width.

\begin{proposition}
\label{prop:wrat-finite-width}
    For $L \in \infrat$, $\ms{wd}(L) < + \infty$.
\end{proposition}

Thus, Prop.~\ref{prop:wrat-finite-width} and Lem.~\ref{lem:infinite-product-well-defined} imply that $\infrat \subset 2^\infipoms$, that is, $\omega$-rational operations preserve validity. 
In addition, since $\Xi \subset \iipoms$ and by Cor.~\ref{cor:inf-product-interval}, we have:

\begin{proposition}
\label{prop:infrat-in-infiipoms}
    For all $L \in \infrat$, we have $L \subset \infiipoms$.
\end{proposition}

We write $\wrat = \infrat \cap 2^\wiipoms$ for \emph{$\omega$-rational} languages, and $\rat = \infrat \cap 2^\iipoms$ for finite ones. 
Note that $\rat$ is not the class of rational languages defined in~\cite{DBLP:conf/concur/FahrenbergJSZ22} (that we denote $\Drat$ for down-closed rational).
For example, $\set{\loset{a \ibu \\b \ibu}} * \set{\loset{\ibu a \\\ibu b}} = \set{\loset{a\\b}} \in \rat$ as the gluing of a starter and a terminator, but is not in $\Drat$ (because not down-closed). 
We show in Prop.~\ref{prop:finite-rationality} that up to down-closure, both notions are equivalent in the finite case. 

\section{$\omega$-rationality vs.\ $\omega$-regularity}
\label{sec:relation-of-w-languages}

In this section, we explore the connections between $\omega$-rationality and $\omega$-regularity. 
We show that $\omega$-regular languages are $\omega$-rational.
The proof relies on a type of classical $\omega$-automata derived from $\omega$-HDAs called \emph{ST-automata}.
On the other hand, the opposite does not hold.
Indeed, we show that Muller acceptance is more expressive than Büchi acceptance and that there are $\omega$-rational languages that are not Muller $\omega$-regular.

\subsection{ST-automata for $\omega$-HDA}
\label{sec:ST-automata}

An ST-automaton is a plain ($\omega$-)automaton over $\Xi$, with an additional labeling of states, that can produce ($\omega$-)iipomsets if the letters are glued. It especially can mimic the behavior of an ($\omega$-)HDA with a canonical translation.
We use here the most recent definition of these objects, introduced in \cite{PIPI},
which subsumes other variants that
have been used in \cite{DBLP:conf/concur/FahrenbergJSZ22,
	amrane.23.ictac,
	amrane2024languages,
        conf/birthday/AmraneBCFS25,
        AMRANE2025115156,
	DBLP:journals/lites/Fahrenberg22,
	DBLP:conf/adhs/Fahrenberg18}.

We let $\Xib=(\Xib, {\cdot}, \varepsilon)$ denote the free monoid on $\Xib=\Xi$ seen as an alphabet,
using concatenation instead of gluing
and $\varepsilon$ for the empty word (which is different from the letter $\eps \in \Xib$).

An \emph{ST-automaton} $A$ is a finite labeled automaton $(Q,E,I,F,\lambda)$ over the (infinite) alphabet $\Xib$, where $Q$ is the set of states, $E \subseteq Q \times (\Xib \setminus \Id) \times Q$ is the set of transitions, $I\subseteq Q$ are the initial states, $F$ is an acceptance condition and $\lambda : Q \maps \conc$ is a labeling of  states that is coherent with $E$, meaning that for all $(p,\ilo{S}{U}{T},q) \in E,\, \lambda(p) = S$ and $ \lambda(q) = T$.

A \emph{path} $\pi = (q_0,e_0,q_1,e_1,\dots, e_{n-1}, q_n)$ in an ST-automaton $A$
is an alternation of states and transitions such that ${e_i = (q_i,P_i,q_{i+1}) \in E}$.
Its label is $l(\pi) =\id_{\lambda(q_0)}\cdot P_0 \cdot \id_{\lambda(q_1)} \cdot$ $\ldots \cdot P_{n-1} \cdot \id_{\lambda(q_n)}$ seen as a word. 
We say that $\pi$ is accepting if $q_0 \in I$ and $q_n \in F$.
The language of an ST-automaton $A$, denoted $L(A)$, is the set of labels of its accepting paths. The same can be done for $\omega$-paths, which are accepted according to a Büchi or Muller condition.
\begin{remark}
  Path labels of ST-automata are elements of $\Id \cdot (\Xib\cdot\Id)^*$ (or $\Id \cdot (\Xib\cdot \Id)^\omega$ in the infinite case). In particular, the labeling of states and their consideration in the path labels forbid to have the empty word $\varepsilon$ as the label of a path.
  Note also that labels of states are not used twice:
  given $\pi_1 = (q_0,e_0,\dots, e_{i-1}, q_i)$ and $\pi_2 = (q_i,e_i,,\dots, e_{n-1}, q_n)$, we have $l(\pi_1) \cdot l(\pi_2) = \id_{\lambda(q_0)}\cdot P_0 \cdot \id_{\lambda(q_1)} \cdot\ldots\cdot P_{i-1} \cdot \id_{\lambda(q_i)} \cdot P_{i+1} \cdot\ldots\cdot P_{n-1} \cdot \id_{\lambda(q_n)}$.
  We are thus only able to concatenate $l(\pi_1)$ and $l(\pi_2)$ if the end of $\pi_1$ has the same label as the start of $\pi_2$.
  We will consider ($\omega$-)rational expressions of ST-automata as classical ($\omega$-)rational expressions, the only difference being that the concatenation operator behaves as stated above.
\end{remark}
In the following, we build an ST-automaton $ST(X)$ from an ($\omega$-)HDA $X$, following \cite{PIPI}.
The intuition is that each cell (of any dimension) becomes a state, and for each upstep $(p=\delta_A^0(q),{\arrO{A}},q)$ resp.\ downstep $(p,\arrI{B},\delta_B^1(p)=q)$ in $X$, a transition is introduced from the state corresponding to $p$ to the one corresponding to $q$, labeled with $\starter{\ev(q)}{A}$ resp.\ $\terminator{\ev(p)}{B}$), to mimic the behavior of the HDA. More formally, for an HDA $(X,\bot_X,F_X)$, we associate the ST-automaton $ST(X)=(Q,E,I,F,\lambda)$ with:

\begin{itemize}
    \item 
        $Q = X$, $I = \bot_X$, $F =F_X$, $\lambda = \ev$,
    \item
        $E = \set{ (\delta_A^0(q),\starter{\ev(q)}{A},q) \midbar A \subset \ev(q)} \cup \set{(q,\terminator{\ev(q)}{A},\delta_A^1(q)) \midbar A \subset \ev(q)}$.
\end{itemize}
By construction, there is a one-to-one correspondence between tracks in $X$ and paths in $ST(X)$: with a track $\alpha = (q_0,\phi_0,q_1,...)$ of $X$ we associate the path $ST(\alpha)= (q_0,e_0,q_1,...)$ of $ST(X)$ such that
    \begin{itemize}
        \item
            $(q_i,\phi_i,q_{i \+ 1}) = (\delta_{A_i}^0(q_{i \+ 1}),\arrO{A_i},q_{i \+ 1}) \implies A_i \subset \ev(q_{i \+ 1})$, $e_i = (q_i,\starter{\ev(q_{i \+ 1})}{A_i},q_{i \+ 1}) \in E$
        \item
            $(q_i,\phi_i,q_{i \+ 1}) = (q_i,\arrI{A_i},\delta_{A_i}^1(q_i)) \implies A_i \subset \ev(q_i)$, $e_i = (q_i,\terminator{\ev(q_i)}{A_i},q_{i \+ 1}) \in E$
    \end{itemize}
This proves the following lemma:
\begin{lemma}
	\label{lem:equiv-path}
	A track $\alpha$ is accepting in $X$ if and only if $ST(\alpha)$ is accepting in $ST(X)$.
\end{lemma}

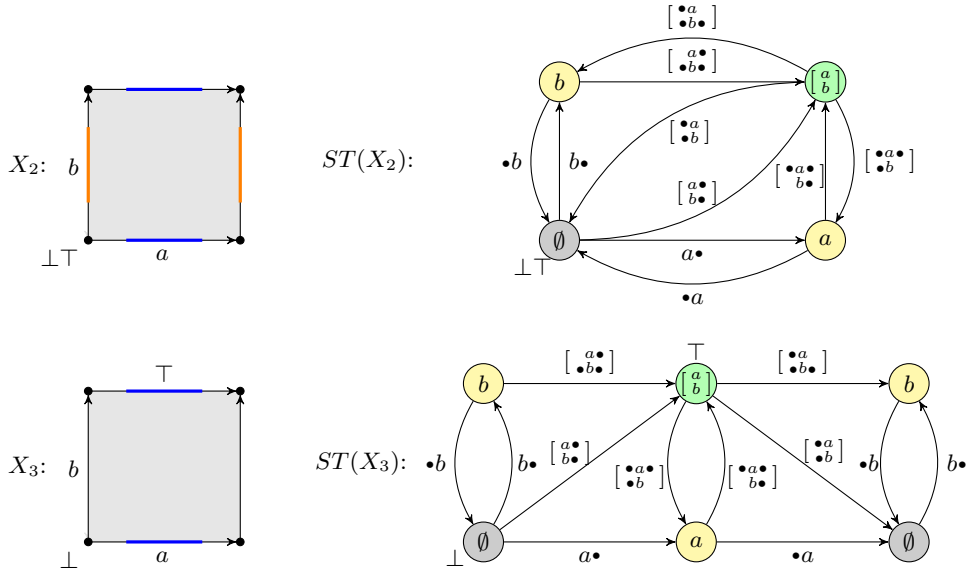
\begin{figure}[!htb]

	\centering
	\begin{tikzpicture}

        \begin{scope}[shift={(0,5.5)},y=1cm]

            \node at (-0.8,1) {$X_2$:};
            
            \filldraw[color=black!10!white] (0,0)--(2,0)--(2,2)--(0,2)--(0,0);
            
            \filldraw (0,0) circle (0.05);
            \filldraw (2,0) circle (0.05);
            \filldraw (0,2) circle (0.05);
            \filldraw (2,2) circle (0.05);
            
            \path (0,0) edge node[below]{$a$}(1.95,0);
            \path (0,0) edge node[left] {$b$} (0,1.95);
            \path (2,0) edge (2,1.95);
            \path (0,2) edge (1.95,2);
    
            \node[below left] at (0,0) {$\bot \top$};
            
            \path (0.5,0) edge[color=blue,-, very thick] (1.5,0);
            \path (0.5,2) edge[color=blue,-, very thick] (1.5,2);
            \path (0,0.5) edge[color=orange,-, very thick] (0,1.5);
            \path (2,0.5) edge[color=orange,-, very thick] (2,1.5);
        
        \end{scope}
        
		\begin{scope}[shift={(6.2,5.5)}, scale=.7]

            \node at (-3.6,1.5) {$ST(X_2)$:};
  
			\node[circle,draw=black,fill=black!20,inner sep=0pt,minimum size=15pt] (i) at (0,0) {$\emptyset$};
			\node[circle,draw=black,fill=yellow!40,inner sep=0pt,minimum size=15pt] (b) at (0,3) {$b$};
			\node[circle,draw=black,fill=yellow!40,inner sep=0pt,minimum size=15pt] (a) at (5,0) {$a$};
			\node[circle,draw=black,fill=green!30,inner sep=0pt,minimum size=15pt] (ab) at (5,3) {$\loset{a\\b}$};
   
			\path (i) edge node[right] {$b \ibu$} (b);
            \path (b) edge[bend right] node[left] {$\ibu b$} (i);
            
			\path (a) edge node[below left=-0.08cm] {$\loset{ \ibu a \ibu \\ \nibu b \ibu }$\!} (ab);
            \path (ab) edge[bend left] node[right] {$\loset{ \ibu a \ibu \\ \ibu b \nibu }$} (a);
   
            \path (i) edge node[below] {$a \ibu$} (a);
            \path (a) edge[bend left] node[below] {$\ibu a$} (i);
            
            \path (b) edge node[above] {$\loset{ \nibu a \ibu \\ \ibu b \ibu }$} (ab);
            \path (ab) edge[bend right] node[above] {$\loset{ \ibu a \nibu \\ \ibu b \ibu }$} (b);

            \path (i) edge[bend right] node[above left=-0.13cm] {$\loset{ a \ibu \\ b \ibu }$} (ab);
            \path (ab) edge[bend right] node[below right=-0.1cm] {$\loset{ \ibu a \\ \ibu b }$} (i);

			\node at (-0.5,-0.5) {$\bot \top$};
   
		\end{scope}

        \begin{scope}[shift={(0,1.5)}]

            \node at (-0.8,1) {$X_3$:};
        
            \filldraw[color=black!10!white] (0,0)--(2,0)--(2,2)--(0,2)--(0,0);
            
            \filldraw (0,0) circle (0.05);
            \filldraw (2,0) circle (0.05);
            \filldraw (0,2) circle (0.05);
            \filldraw (2,2) circle (0.05);
            
            \path (0,0) edge node[below]{$a$}(1.95,0);
            \path (0,0) edge node[left] {$b$} (0,1.95);
            \path (2,0) edge (2,1.95);
            \path (0,2) edge (1.95,2);
    
            \node[below left] at (0,0) {$\bot$};
            \node[above] at (1,2) {$\top$};
            
            \path (0.5,0) edge[color=blue,-, very thick] (1.5,0);
            \path (0.5,2) edge[color=blue,-, very thick] (1.5,2);
        
        \end{scope}
        
		\begin{scope}[shift={(5.2,1.5)}, scale=.7]

            \node at (-2.3,1.5) {$ST(X_3)$:};
  
			\node[circle,draw=black,fill=black!20,inner sep=0pt,minimum size=15pt] (i) at (0,0) {$\emptyset$};
			\node[circle,draw=black,fill=black!20,inner sep=0pt,minimum size=15pt] (e) at (8,0) {$\emptyset$};S
			\node[circle,draw=black,fill=yellow!40,inner sep=0pt,minimum size=15pt] (b1) at (0,3) {$b$};
			\node[circle,draw=black,fill=yellow!40,inner sep=0pt,minimum size=15pt] (a) at (4,0) {$a$};
			\node[circle,draw=black,fill=yellow!40,inner sep=0pt,minimum size=15pt] (b2) at (8,3) {$b$};
			\node[circle,draw=black,fill=green!30,inner sep=0pt,minimum size=15pt] (ab) at (4,3) {$\loset{a\\b}$};
   
			\path (i) edge[bend right] node[right=-0.05cm] {$b \ibu$} (b1);
            \path (b1) edge[bend right] node[left] {$\ibu b$} (i);
			\path (a) edge[bend right] node[below right=-0.1cm] {$\loset{ \ibu a \ibu \\ \nibu b \ibu }$} (ab);
            \path (ab) edge[bend right] node[below left=-0.1cm] {$\loset{ \ibu a \ibu \\ \ibu b \nibu }$} (a);
			\path (e) edge[bend right] node[right] {$b \ibu$} (b2);
            \path (b2) edge[bend right] node[left=-0.1cm] {$\ibu b$} (e);
            
            \path (i) edge node[below] {$a \ibu$} (a);
            \path (b1) edge node[above] {$\loset{ \nibu a \ibu \\ \ibu b \ibu }$} (ab);
            \path (a) edge node[below] {$\ibu a$} (e);
            \path (ab) edge node[above] {$\loset{ \ibu a \nibu \\ \ibu b \ibu }$} (b2);

            \path (i) edge node[above left=-0.15cm] {$\loset{ a \ibu \\ b \ibu }$} (ab);
            \path (ab) edge node[above right=-0.13cm] {$\loset{ \ibu a \\ \ibu b }$} (e);

			\node at (-0.3,-0.3) {$\bot\ \ \ $};
			\node at (4,3.6) {$\top$};
   
		\end{scope}
		
	\end{tikzpicture}
 
	\caption{ST-automata of $X_2$ and $X_3$ of Ex.~\ref{ex:buchi}. Cells of dimension 0 are in gray, dimension 1 in yellow, and dimension 2 in green.}
	\label{fig:ST-example}
\end{figure}

Since the labeling $\lambda$ is coherent with $E$, a word produced by an ST-automaton $A$ is \emph{coherent} in the sense that if $P_1\cdot P_2\cdot \ldots$ is the label of a path in $A$, then $T_{P_i} \cong S_{P_{i\+1}}$ for all $i \geq 0$ (and thus it can be seen as the gluing of starters and terminators). To be more precise, the set of coherent finite words over $\Xib$ is denoted $\coh(\Xib) = \set{U_1\cdot \ldots \cdot U_n \midbar U_i \in \Xib,\ T_{U_i} \cong S_{U_{i\+1}}}$.

For $\omega$-words we denote by $ \wcoh(\Xib)$ the set of coherent infinite words over $\Xib$, defined by 
$${\wcoh(\Xib) = \set{U_1\cdot U_2 \cdot \ldots \midbar U_i \in \Xib,\ \#\set{i : U_i \not\in \Id} = + \infty \text{ and } T_{U_i} = S_{U_{i\+1}}}}$$
and as usual we write $\infcoh(\Xib) = \coh(\Xib) \cup \wcoh(\Xib)$.

\begin{lemma}
\label{lem:ST-coh}
    Let $A$ be an ST-automaton.
    \begin{itemize}
        \item
            If $A$ accepts finite words then $L(A) \subset \coh(\Xib)$.
        \item
            If $A$ accepts infinite words then $L(A) \subset \wcoh(\Xib)$.
    \end{itemize}
\end{lemma}

Thus, one can ``glue'' words accepted by an ST-automaton to obtain an ($\omega$-)iipomset.
Let $\Psi \colon \infcoh(\Xib) \to \infipoms$ be the \emph{gluing function} defined by $\Psi(P_1 \cdot P_2 \cdot \dots) = P_1 * P_2 * \dots$.
For a set $L \subset \infcoh(\Xib)$, we also define $\Psi(L) = \set{\Psi(P) \midbar P \in L}$.
In \cite{PIPI}, the authors show the following for ST-automata derived from HDAs over finite iipomsets:
\begin{proposition}[\cite{PIPI}]
\label{prop:eq-ST-HDA}
    For any HDA $X$, $L(X) = \Psi(L(ST(X)))$.
\end{proposition}

We show  that the proposition holds also for $\omega$-HDAs. 
First, since the labeling of states is coherent with the transitions, we have the following by applying definitions.

\begin{lemma}
\label{lem:equiv-evalutation}
    For all $\omega$-tracks $\alpha$ of $X$, we have $\ev(\alpha) = \Psi(l(ST(\alpha))$.
\end{lemma}

We can now show the $\omega$-equivalent of Prop.~\ref{prop:eq-ST-HDA}, by combining Lem.~\ref{lem:equiv-path} and Lem.~\ref{lem:equiv-evalutation}:

\begin{proposition}
\label{prop:eq-ST-wHDA}
    For any (Büchi or Muller) $\omega$-HDA $X$, $L(X) = \Psi(L(ST(X)))$.
\end{proposition}

\subsection{$\omega$-HDAs are $\omega$-rational}
\label{sec:eq-ST-wHDA}

The formalism of ST-automata provides a strong tool for HDAs. 
It allows  to use classical theorems of $\omega$-automata. 
This helps in particular to show that $\omega$-regular languages are $\omega$-rational 	as we will see in this section.

Indeed, the language of any ST-automaton $A$ is also the one of a rational expression $e_A$, using the operations $+,.\,,^*$ and letters $U_i \in \Xib$. 
Similarly, given an ST-automaton $A'$ over infinite words, one can effectively build an equivalent $\omega$-rational expression $e_{A'}$ \cite{mcnaughton1966testing}. 

A ($\omega$-)rational expression is said \emph{positive} if it has no occurrences of $^*$ and $\varepsilon$.
Since ST-automata do not accept $\varepsilon$, ST-automata languages can be expressed by positive expressions:

\begin{lemma}
\label{lem:positive-expression}
    Any rational expression $e$ of an ST-automaton can be transformed into a positive rational expression $\overline{e}$ such that $\mc{L}(e)=\mc{L}(\overline{e})$.
\end{lemma}

We define inductively a function $\Psi'$ transforming a positive rational expression over  $\Xib$  into a set of ipomsets in the usual way by replacing $\cdot$ by $*$.
We directly have:

\begin{lemma}
\label{lem:psi'-in-wrat}
    Let $e(U_1,...,U_n)$ be a positive rational expression over $\Xib$.
    \begin{itemize}
        \item
            If $e$ is rational, then $\Psi'(e) \in  \rat$.
        \item
            If $e$ is $\omega$-rational, then $\Psi'(e) \in  \wrat$.
    \end{itemize}
\end{lemma}

In addition, we can show by induction that $\Psi'$ preserves languages. 
\begin{lemma}
\label{lem:psi'-psi}
    If $e$ is an ($\omega$-)rational expression over $\Xib$, then $\Psi'(e) = \Psi(\mc{L}(e))$
\end{lemma}

Thus, languages of ST-automata are rational:
\begin{lemma}
\label{lem:ST-rat}
    Let $A$ be an ST-automaton.
    \begin{itemize}
        \item
            If $A$ accepts finite words, then $\Psi(L(A)) \in \rat$.
        \item
            If $A$ accepts infinite words, then $\Psi(L(A)) \in \wrat$.
    \end{itemize}
\end{lemma}

\begin{theorem}
    \label{th:w-rationality}
    Büchi and Muller $\omega$-regular languages are $\omega$-rational.
\end{theorem}

\begin{proof}
    Let $(X,\bot_X,F_X)$ be a Büchi or Muller $\omega$-HDA. By definition of $\omega$-HDAs, $X$ is finite, hence so is its ST-automaton $ST(X)$. By Prop.~\ref{prop:eq-ST-wHDA}, $L(X) = \Psi(L(ST(X)))$, and by Lem.~\ref{lem:ST-rat}, $\Psi(L(ST(X))) \in \wrat$, thus all $\omega$-regular languages are $\omega$-rational.
\end{proof}

As a corollary of the above and Prop.~\ref{prop:infrat-in-infiipoms} we have:

\begin{corollary}
\label{cor:Buchi-Muller-finite-width}
    If $L$ is an $\omega$-regular language, then $\ms{wd}(L) < +\infty$ and $L \subset \wiipoms$.
\end{corollary}


\begin{example}
\label{ex:wrat-languages}
    In Fig.~\ref{fig:ST-example} are represented two $\omega$-HDAs with their corresponding ST-automata. From this we can directly compute their language in $\Xib^\omega$ (we omit the letters in $\Id$), and then use $\Psi$ to compute the languages of the original Büchi $\omega$-HDAs:
    \begin{align*}
            L(ST(X_2)) =& \Big(
            a \ibu \cdot \ibu a + b \ibu \cdot \ibu b
            + 
            (a\ibu \cdot \loset{\ibu a\ibu \\ \nibu b\ibu} + \loset{a \ibu \\ b \ibu} + b\ibu \cdot \loset{\nibu a\ibu \\ \ibu b\ibu})
            \cdot
            (\loset{\ibu a\ibu \\ \ibu b\nibu} \cdot\loset{\ibu a\ibu \\ \nibu b\ibu} + \loset{\ibu a\nibu \\ \ibu b\ibu} \cdot\loset{\nibu a\ibu \\ \ibu b\ibu})^*
            \\
            &\cdot(\loset{\ibu a \ibu \\ \ibu b \nibu} \cdot \ibu a + \loset{\ibu a \\ \ibu b} + \loset{\ibu a \nibu \\ \ibu b \ibu} \cdot \ibu b)
            \Big)^\omega\\
             L(ST(X_3)) =&
            (b \ibu \cdot \ibu b)^*
            \cdot
            (a \ibu \cdot \loset{\ibu a \ibu \\ \nibu b \ibu} + b \ibu \cdot \loset{\nibu a \ibu \\ \ibu b \ibu} + \loset{ a \ibu \\ b \ibu})
            \cdot
            (\loset{\ibu a \ibu \\ \ibu b \nibu} \cdot \loset{\ibu a \ibu \\ \nibu b \ibu})^\omega
    \end{align*}
    Languages in $\wcoh(\Xib)$ can be turned into the language of the original HDA using the function $\Psi$ and Prop.~\ref{prop:eq-ST-HDA}. We just replace the occurrences of $A^*$ by $A^+ + \varepsilon$ and develop the expression to avoid the word $\varepsilon$, then remove multiple occurrences of the same ipomset in unions:
    \begin{itemize}
        \item 
            $L(X_2) = \Psi(L(ST(X_2))) = (a + b + \loset{ a \\ b } +\loset{ a \ibu \\ b \ibu } *  (\loset{\ibu a \ibu \\\ibu bb \ibu } + \loset{\ibu aa \ibu \\\ibu b \ibu })^+ * \loset{ \ibu a \\ \ibu b})^\omega$
        \item
            $L(X_3) = \Psi(L(ST(X_3))) = b^+ * \loset{ a \ibu \\ b \ibu } * \loset{ \ibu a \ibu \\ \ibu bb \ibu }^\omega + \loset{ a \ibu \\ b \ibu } * \loset{ \ibu a \ibu \\ \ibu bb \ibu }^\omega$
    \end{itemize}
\end{example}

As we saw in the previous section, our definition of rational operations is different from the original one \cite{DBLP:conf/concur/FahrenbergJSZ22} since  we do not take subsumption closure into account.
Nevertheless:
\begin{proposition}
\label{prop:finite-rationality}
    $\rat\down = \Drat$.
\end{proposition}
\begin{proof}
	The inclusion $\rat\down \subseteq \Drat$ is done by induction.
	The other follows from Lem.~\ref{lem:ST-rat} and the fact that languages of $\Drat$ are closed under subsumption (see App.~\ref{app:proof-relation-of-w-language} for details).
\end{proof}

\subsection{Muller vs Büchi vs $\omega$-rational}

In the classical $\omega$-theory, rational languages are as expressive as Muller and Büchi automata.
For $\omega$-HDA, these are no longer true. We describe and prove here some differences:

\begin{theorem}
    \label{prop:Muller-vs-Büchi}
    The Muller $\omega$-regular class is strictly bigger than the Büchi $\omega$-regular class.
\end{theorem}

\begin{proof}
    As for the classical $\omega$-theory, Büchi $\omega$-regular languages are Muller $\omega$-regular.
    
    However, the converse is false. Take the Muller $\omega$-HDA $X_4$ of Ex.~\ref{ex:muller}. The only cycle (up to shift) using only states of $F_{X_4}$ is $(q_0, \arrO{ab},q_{ab}, \arrI{ab}, q_0, ...)$, so the language of $X_4$ is such that every accepted $\omega$-ipomset ends by $\loset{a \\ b}^\omega$. Assume that there is a Büchi $\omega$-HDA $Y$ that recognizes $L(X_4)$. Let $\alpha$ be an accepting $\omega$-track of $Y$, which exists because $L(X_4) \neq \emptyset$. We have $\ev(\alpha)= A*\loset{a \\ b}^\omega$ for some $A \in \iipoms$. So $\alpha$ goes through $Y[\loset{a \\ b}]$ an infinite number of times. However $Y$ is finite, so there is $p \in Y[\loset{a \\ b}]$ such that $\alpha$ goes through $p$ an infinite number of times. Thus there is a subtrack $\alpha_p := (\delta^0_{ab}(p), \arrO{ab}, p, \arrI{ab}, \delta^1_{ab}(p))$ which appears an infinite number of times in $\alpha$. (It may be a track equivalent to it, but as there is a finite number of them, it suffices to choose one of them that is used an infinite number of times.) As $Y$ is a precubical set, we can define the track $\alpha'_p = (\delta^0_{ab}(p) ,\arrO{a}, \delta^0_b(p), \arrI{a}, \delta^1_a(\delta^0_b(p)), \arrO{b}, \delta^1_a(p), \arrI{b}, \delta^1_{ab}(p))$ with label $ab$, and then replace one in two occurrences of $\alpha_p$ by $\alpha'_p$. The cells visited an infinite number of times by $\alpha$ are still visited an infinite number of times, hence the new $\omega$-track is also accepting. However, its label does not end by $\loset{a \\b}^\omega$ because $ab$ appears an infinite number of times. So $L(X_4)$ is not Büchi $\omega$-regular.
\end{proof}

\begin{theorem}
    \label{prop:wrational-not-muller}
    Some languages of $\wrat$ cannot be expressed by Muller $\omega$-HDA.
\end{theorem}

\begin{proof}
    Let $L_r = (\loset{a \ibu \\ b\ibu} * \loset{\ibu a \\ \ibu b})^\omega = \loset{a \\ b}^\omega\in \wrat$. Suppose that $(X,\bot_X,F_X)$ is a Muller $\omega$-HDA such that $L(X)= L_r$. An accepting $\omega$-track $\alpha$ (with label $\ev(\alpha) = \loset{a \\ b}^\omega$) in $X$ must start by a track (equivalent to) $\alpha_p := (\delta^0_{ab}(p), \arrO{ab}, p, \arrI{ab}, \delta^1_{ab}(p))$ with $p \in X[\loset{a\\b}]$. As $X$ is a precubical set, $\alpha'_p = (\delta^0_{ab}(p) ,\arrO{a}, \delta^0_b(p), \arrI{a}, \delta^1_a(\delta^0_b(p)), \arrO{b}, \delta^1_a(p), \arrI{b}, \delta^1_{ab}(p))$ is well-defined. If $\beta$ is such that $\alpha = \alpha_p * \beta$, then $\alpha' = \alpha'_p * \beta$ is a track in $X$ which is accepting (as Muller acceptance conditions only care about states seen an infinite number of times). Then $\ev(\alpha') = ab\loset{a\\b}^\omega \in L(X)$ but $ab\loset{a\\b}^\omega \not\in L_r$, and we conclude by contradiction.
\end{proof}

Note that defining $\omega$-rational operations with subsumption closure (following~\cite{DBLP:conf/concur/FahrenbergJSZ22}) would have made
$\omega$-rational languages \emph{less} expressive than $\omega$-HDAs.
For example, as $\{\loset{a \\ b}\}\down = \set{ab, ba, \loset{a \\ b}}$, the language $L(X_4)$ of Ex.~\ref{ex:muller} would then not be $\omega$-rational.

As $\wrat$ is more expressive than Büchi and Muller $\omega$-HDAs, it is natural to inquire about the form of languages of $\wrat$ with the same expressiveness. 
As a first step, define a \emph{locally down-closed lasso} (\emph{ldl}) $\omega$-language to be a language of the form $\bigcup\nolimits_{i \leq n} M_i\down * ((R_i^+)\down)^\omega$ for $n \in \Nat$ and  $M_i, R_i \in \rat$ for all $i \leq n$.

\begin{theorem}
    \label{prop:buchi-wform}
    Any Büchi regular $\omega$-language is a locally down-closed lasso $\omega$-language.
\end{theorem}

Thus, $\omega$-HDAs allow local subsumptions in some finite factors coming from $M_i$ and $R_i^+$. 
Using $R_i^+$ instead of $R_i$ more accurately reflects the behaviour of Büchi $\omega$-HDAs.
Indeed, for $R = \set{ \loset{\ibu a \\ \ibu b} * \loset{a \ibu \\ b \ibu}}$ and $M= \loset{a \ibu \\ b \ibu}$, $M$ and $R$ are down-closed, but $M * R^\omega = M\down * (R\down)^\omega = \set{ \loset{a \\ b}^\omega}$ which cannot be recognized by a Büchi $\omega$-HDA as seen in the proof of Th.~\ref{prop:Muller-vs-Büchi}.

\begin{example}
\label{ex:buchi-wform}
  We can use again Ex.~\ref{ex:wrat-languages}, and take $L(X_3)$ which is Büchi $\omega$-regular:
    \begin{align*}
            L(X_3)  = b^+ * \loset{ a \ibu \\ b \ibu } * \loset{ \ibu a \ibu \\ \ibu bb \ibu }^\omega + \loset{ a \ibu \\ b \ibu } * \loset{ \ibu a \ibu \\ \ibu bb \ibu }^\omega = (b^+ + \eps) * \loset{ a \ibu \\ b \ibu } * (\loset{ \ibu a \ibu \\ \ibu bb \ibu }^+)^\omega 
    \end{align*}
    For $M =  (b^+ + \eps) * \loset{ a \ibu \\ b \ibu }$ and $R =\loset{ \ibu a \ibu \\ \ibu bb \ibu }$, we have $L(X_3) = M\down * ((R^+)\down)^\omega$.
\end{example}

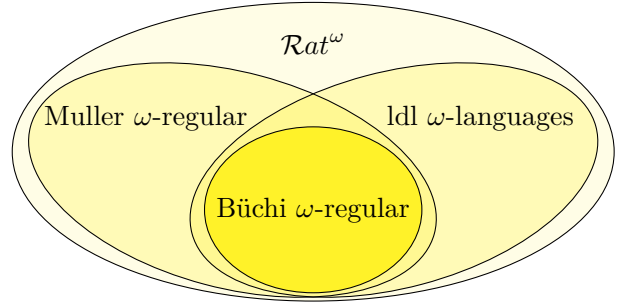
\begin{figure}[!h]

    \def\rot{15}
    \centering
    
    \begin{tikzpicture}[scale=1.1]
        \tikzstyle{every node}=[font=\large]

        \fill[fill = yellow!50, nearly transparent]  (0,0.6) ellipse (3.6cm and 1.8cm);
        \draw (0,0.6) ellipse (3.6cm and 1.8cm);

        \fill[fill = yellow!90, nearly transparent, rotate=-\rot]  (-1,0) ellipse (2.5cm and 1.3cm);
        \fill[fill = yellow!90, nearly transparent, rotate=\rot]  (1,0) ellipse (2.5cm and 1.3cm);
        \draw[rotate=-\rot] (-1,0) ellipse (2.5cm and 1.3cm);
        \draw[rotate=\rot] (1,0) ellipse (2.5cm and 1.3cm);

        \fill[fill = yellow!90]  (0,-0.1) ellipse (1.3cm and 1.0cm);
        \draw (0,-0.1) ellipse (1.3cm and 1.0cm);

        \node at (0,1.9) {$\wrat$};
        \node at (-2,1) {Muller $\omega$-regular};
        \node at (2,1) {ldl $\omega$-languages};
        \node at (0,-0.1) {Büchi $\omega$-regular};
        
    \end{tikzpicture}

\caption{Inclusion of classes of $\omega$-languages, with Büchi $\subsetneq$ Muller $\subsetneq \wrat$ }
\label{fig:language-diagram}
\end{figure}

\section{Conclusion and future work}

We have defined interval $\omega$-pomsets with interfaces ($\omega$-iipomsets),
shown that isomorphisms between them are unique, and that they admit a unique sparse decomposition.
Then, we have introduced $\omega$-higher-dimensional automata ($\omega$-HDAs) over $\omega$-iipomsets as a generalization of HDAs over finite iipomsets.
We have studied their languages under Büchi and Muller acceptance conditions and demonstrated that Muller is more expressive than Büchi. 

Unlike in the finite case, languages of $\omega$-HDAs are not closed under subsumption. 
In pursuit of a Kleene-like theorem, we have adapted the rational operations on finite iipomsets to disregard subsumption,
showing that up to subsumption, the notion of rational language remains unchanged under this modification.
We have also introduced a non-nested $\omega$-iteration to these operations and defined $\omega$-rational languages over $\omega$-iipomsets. 
However, we show that this natural class is bigger than the one of $\omega$-iipomset languages recognized by $\omega$-HDAs.
The diagram in Fig.~\ref{fig:language-diagram} summarizes our results about $\omega$-HDAs and $\omega$-rational languages.
In the diagram, the inclusion of Büchi $\omega$-regular languages into locally down-closed lasso $\omega$-languages
is shown as strict (cf.~Th.~\ref{prop:buchi-wform}); but we conjecture it to be an equivalence.

The challenge in establishing a Kleene-like theorem arises from the precubical identities
that force some factors to be subsumption-closed.
A solution may be found by passing to the \emph{partial} HDAs of \cite{DBLP:conf/calco/FahrenbergL15, DBLP:conf/fossacs/Dubut19}.
Intuitively, partial HDAs relax the face maps to be partial functions,
satisfying precubical inclusions rather than precubical identities.
However, while the language theory of HDAs has seen much recent progress, partial HDAs remain largely unexplored.
The use of such objects is particularly relevant, as with our adapted rational operations, we have shown that the language of any ST-automaton is ($\omega$-)rational, which is \emph{not} the case for the subsumption-closed rational operations of \cite{DBLP:conf/concur/FahrenbergJSZ22}.
This notably implies that ``our'' rational languages are closed under bounded-width complement, whereas the original definition is only closed under bounded pseudocomplement \cite{amrane.23.ictac}.

This work represents a first step toward developing a theory of $\omega$-HDAs.
Such a theory should be well-suited for modeling non-terminating concurrent systems with dependence and independence relations.
In particular it would avoid some problems with state-space explosion; for example, an infinite execution where event $a$ must precede $b$ is modeled as the subsumption closure of the pomset with $a < b$ and all other events occurring in parallel, instead of considering all interleavings separately.
Developing a logical characterization, similar to the finite case \cite{amrane2024logic}, would also be of interest.

\bibliographystyle{plainurl}
\bibliography{ref.bib}
\newpage\appendix
\section{Complementary proofs of Section~\ref{sec:infinite-ipomsets}}
\label{app:proof-infinite-ipomset}

\begin{proof}[Proof of Proposition~\ref{prop:unicity-iso}]

Let $P$ be an $\omega$-ipomset and define $P_0$ as the set of $<_P$-minimal elements of $P$ and $P_{i+1}$ the set of $<_P$-minimal elements of $P \setminus \{P_0,\dots,P_i\}$ for all $i < \omega$.
Note that each $P_i$ is an $<_P$-antichain and linearly ordered by $\evord_P$.
Define the binary relation $\mc{R}_P$ by $x \mc{R}_P y$ iff ($x \in P_i$, $y \in P_j$ and $i <j$) or $(x,y \in P_i$ and $x \evord_P y$).
Then $(P,\mc{R}_P)$ is a strict well-order.
In addition,  an isomorphism between $\omega$-ipomsets $P$ and $Q$	is an isomorphism between $(P,\mc{R}_P)$ and $(Q,\mc{R}_Q)$.
We conclude by noting that well-ordering isomorphisms are unique (see~\cite[Lem.~6.2]{Kunen80}).
\end{proof}


\begin{proof}[Proof of Lemma~\ref{lem:infinite-product-well-defined}]
	
	Let $(P_i)_{i \in \Nat} \in \ipoms^\Nat$ such that $S_{P_i} \cong T_{P_{i\+1}}$, $I = \set{i \in \Nat \midbar P_i \not\in \Id}$ is infinite and there is $m \in \Nat$ such that $\ms{wd}(P_i) \leq m$. We want to show that it is a valid $\omega$-ipomset.
	It is routine to check that $P = P_0 * P_1 * \cdots$ is well-defined, however it can be a finite ipomset (for example for a constant sequence in $\Id$), or can define a not valid $\omega$-ipomset.  
	\begin{itemize}
		\item Infinity:
		By contradiction, suppose $P$ is finite. If so, there is an index $i_0$ such that, for $i > i_0$, $P_i$ does not start any event. Let $Q = P_0 * ... * P_{i_0} \in \ipoms$. For $i \in I$, $P_i \not\in \Id$ so it has to start or finish at least one event. Let $I' = I \cap \set{i \in \Nat \midbar i > i_0}$, as $I$ is infinite, $I'$ is infinite. For $ i \in I'$, $P_i$ has to end at least one event of $Q$. But $|Q|$ is finite and $I'$ is infinite, contradiction.
		\item Finite past:
		Let $z \in P$, there is $j \sth z \in P_j$. By definition of $<_P$, if $y <_P z$ then there is $i \leq j$ such that $y \in P_i$. So the $<_P$-predecessors of $z$ are in $(\bigcup_{i \leq j} (P_i,i) )_{/x \sim f(x)}$ which is finite.
		\item No infinite $<_P$-antichain:
		Let $A$ be an $<_P$-antichain, then there is $i$ such that $A \subset P_i$. So $|A| \leq \ms{wd}(P_i) \leq m$, so there is no infinite $<_P$-antichain.
	\end{itemize}
	Thus, $P$ is a valid $\omega$-ipomset.
\end{proof}

The following lemma is needed for the proof of Prop.~\ref{prop:eq-wiipoms} (see \cite{DBLP:journals/corr/abs-2503-07881} for a proof):
\begin{lemma}
	\label{lem:<<-linear}
	Let $P \in \wipoms$ and $E$ be its set of maximal $<_P$-antichains. If $\ll$ is a linear order on $E$, then $(E,\ll) \cong (\Nat,<)$.
\end{lemma}

\begin{proof}[Proof of Proposition~\ref{prop:eq-wiipoms}]
	We want to prove the equivalence between (1) $P$ is an interval $\omega$-ipomset, (2) $P$ has an interval representation and (3) the order $\ll$ on maximal antichains is linear.
	
	\begin{itemize}
		\item (2) $\implies$ (1) :
		Let $w,x,y,z \in P$ with $w <_P y$ and $x <_P z$. By (2) we take an interval representation $b,e$, then $e(w) < b(y)$, $e(x) < b(z)$. Suppose that $w \nless_P z$, by (2) $e(w) \nless b(z)$, as $<$ is linear on $\Nat \cup \set{+\infty}$, $b(z) \leq e(w)$, and thus $e(x) < b(y)$ so by (2), $x <_P y$. Thus, $w <_P z$ or $x <_P y$.        
		\item (1) $\implies$ (3) :
		Let $A$ and $B$ be two different maximal antichains and suppose that $A \not\ll B$ and $B \not\ll A$. Then there is $(x,y) \in A \times B$ such that $x < y $ and  $(x',y') \in A \times B$ such that $y' < x'$. By (1)  we have $x < x'$ or $y' < y$, but $A$ and $B$ are $<$-antichains, contradiction.
		\item (3) $\implies$ (2) :
		By Lem.~\ref{lem:<<-linear}, we can order the set $E$ of maximal antichains of $P$ such that $E = \set{ A_0,\, A_1,\, A_2,\,  ...}$ and $A_0 \ll A_1 \ll A_2 \ll ...$. We define $b(x) = \inf{i \midbar x \in X_i}$ and $e(x) = \sup{i \midbar x \in A_i}$. By definition we have $b(x) \leq e(x)$. Let $x,y \in P \sth x <_P y$, then for all $i < j$ if $x \in A_j$, we have $A_i \ll A_j$ so $y \not\in A_i$, and $y \not\in A_j$ (otherwise $X_j$ is not an antichain), thus $e(x) < b(y)$.  Let $x,y \in P \sth x \not<_P y$. Either $y <_P x$, then by the previous point $e(y) < b(x)$ and so $b(y) < e(x)$, or we have $y \not<_P x$ thus $\set{x,y}$ is include in a maximal antichain of $P$, then there is $i$ such that $x,y \in A_i$ and so $b(y) \leq i \leq e(x)$. Thus $x <_P y \equivalent e(x) < b(y)$.
	\end{itemize}
	So we have (1) $\equivalent$ (2) $\equivalent$ (3).
\end{proof}

Before proving Th.~\ref{th:sparse-dec} we need lemma that characterize the nature (starter versus terminator) of the first element of the decomposition (Lem.~\ref{lem:starter-vs-terminator}), and one that characterise it solely from the $P$ considered (Lem.~\ref{lem:def-start}). We define $P^m$ the $<_P$-minimal elements of $P$ (which can be seen as a conclist with the induced event order of $P$) and $P^s := \set{p\in P \midbar \forall q \in P - P^m,\, p < q} \subset P^m$.

\begin{lemma}
	\label{lem:po=pm}
	Let $ P = P_0 * P_1 * \dots$ be a sparse decomposition, then $P_0 = P^m$ (seen as conclists).
\end{lemma}

\begin{proof}[Proof of Lemma~\ref{lem:po=pm}]
	By definition, we already have that $P_0 \subset P^m$. Let's show that $P - P_0 \subset P - P^m$ (which is equivalent to $P^m \subset P_0$). Let $x \in P - P_0$, there is $i \geq 1 \sth x \in  P_i - S_{P_i}$.
	\begin{itemize}
		\item If $P_0$ is a (proper) terminator:
		it exists $y \in P_0 - T_{P_0}$ and thus $y < x$
		\item  If $P_0$ is a starter:
		then $P_1$ is a (proper) terminator so it exists $y \in P_1 - T_{P_1}$, and $i \geq 2$ (because $P_1 = S_{P_1}$), thus $y < x$
	\end{itemize}
	In both cases there is an element before $x$ so $x \not\in P^m$, so $P - P_0 \subset P - P^m$.
\end{proof}

\begin{lemma}
	\label{lem:starter-vs-terminator}
	Let $P = P_0 * P_1 * \dots$ be a sparse decomposition, then $P_0$ is a starter iff $S_P \subsetneq P^m$ (equivalently, $P_0$ is a terminator iff $S_P = P^m$).
\end{lemma}

\begin{proof}[Proof of Lemma~\ref{lem:starter-vs-terminator}]
	Using Lem.~\ref{lem:po=pm}, $P_0 = P^m$, and by definition $S_{P_0} = S_P$. Then we have that $P_0$ is a starter iff $ S_P = S_{P_0} \subsetneq P_0 = P^m$.
\end{proof}

\begin{lemma}
	\label{lem:def-start}
	Let $P = P_0 * P_1 * \dots$ be a sparse decomposition:
	\begin{itemize}
		\item
		If $P_0$ is a starter, $P_0 =\ \starter{P^m}{P^m - S_P}$
		\item
		If $P_0$ is a terminator, $P_0 = \terminator{P^m}{P^s}$
	\end{itemize}
\end{lemma}

\begin{proof}[Proof of Lemma~\ref{lem:def-start}]
	Using Lem.~\ref{lem:po=pm}, $P_0 = P^m$:
	\begin{itemize}
		\item
		If $P_0$ is a starter, by Lem.~\ref{lem:starter-vs-terminator} $S_P = S_{P_0} \subsetneq P_0 = P^m$ and so $P_0 = \starter{P_0}{P_0 - S_{P_0}} = \starter{P^m}{P^m - S_P}$.
		\item 
		If $P_0$ is a terminator, we want to show that $P^s = P^m - T_{P_0}$, thus $P_0 = \terminator{P^m}{P^s}$. First let $ p \in T_{P_0}$, then $ p \in S_{P_1} \subset P_1$. $P_1$ is a (proper) starter so there is $q \in P_1 -  S_{P_1}$ so $q \not\in P_0 = P^m$ and $p \not< q$ so $p \not\in P^s$, thus $P^s \subset P^m - T_{P_0}$. Then let $ p \in P^m - T_{P_0} = P_0 - T_{P_0}$ and $q \in P - P^m = P - P_0$, thus there is $i \geq 1 \sth q \in P_i - S_{P_i}$, so $p < q$. Thus $p \in P^s$ and $P^m - T_{P_0} \subset P^s$.
	\end{itemize}
	Which concludes the proof.
\end{proof}

We can now prove the fundamental decomposition of $\omega$-iipomsets:

\begin{proof}[Proof of Theorem~\ref{th:sparse-dec}]
	
	Let $P \in \wiipoms$, by Lem.~\ref{lem:<<-linear}, as $\ll$ is linear by Prop.~\ref{prop:eq-wiipoms}, we have that $X_0 \ll X_1 \ll \dots$ . For all $i \in \Nat$, we can define $E_i = X_i - X_{i\+1}$ and $B_i = X_{i\+1} - X_i$. As $X_i$ are maximal antichains, the sets are not empty and $\forall (x,y) \in E_i \times B_i,\, x <_P y$. We then define $U_{2i} = \terminator{X_i}{E_i}$ and $U_{2i\+1} =\ \starter{X_{i\+1}}{B_i}$ which are proper (the event order on $U_{2i}$ and $U_{2i\+1}$ is the one induced by $\evord_P$). If $S_P \subsetneq X_0$, we shift all indices by 1 and redefine $U_0 = \starter{X_0}{X_0 - S_P}$. Thus $U_0*U_1*U_2* \dots$ is a sparse decomposition, and it is routine to check that $P = U_0 * U_1 * U_2 \dots$, concluding of the existence of a sparse decomposition.
	
	Let $P = P_0 * P_1 * \dots = Q_0 * Q_1 * \dots$ be two different sparse decompositions. Without loss of generality, we can delete the common prefix and consider that $P_0 \neq Q_0$. If $S_P \subsetneq P^m$, by Lem.~\ref{lem:starter-vs-terminator} $P_0$ and $Q_0$ are (proper) starters, and by Lem.~\ref{lem:def-start}, $P_0 = \starter{P^m}{P^m - S_P} = Q_0$. If $S_P = P^m$, by Lem.~\ref{lem:starter-vs-terminator} $P_0$ and $Q_0$ are (proper) terminators, and by Lem.~\ref{lem:def-start}, $P_0 = \terminator{P^m}{P^s} = Q_0$, so contradiction.
\end{proof}

\begin{proof}[Proof of Propoposition~\ref{prop:seq-wiipoms}]
	We show the equivalence between (1) $P$ is an interval $\omega$-ipomset, (4) all prefixes of $P$ are interval ipomsets and $\pref{P}$ is infinite, and (5) an infinite number of prefixes of $P$ are interval ipomsets.
	
	\begin{itemize}
		\item (1)  $\implies$ (4) :
		Suppose for all $w,x,y,z \in P$ where $w < y$ and $x < z$, then $w < z$ or $x < y$. Then for all $A \in \pref{P}$, we have the same property for all  $w,x,y,z \in A \subset P$. Moreover, $\pref{P}$ is infinite as a direct consequence of the existence part of Th.~\ref{th:sparse-dec}.
		\item (4) $\implies$ (5) :
		Trivial.
		\item (5) $\implies$ (1) :
		Let $(A_i)_{i \in  I}$ be the infinite set of interval ipomsets in $\pref{P}$. Let $w,x,y,z \in P$ where $w < y$, $x < z$. By Lem.~\ref{lem:pref-exist}, there is a prefix $A_i$ with $\set{w,x,y,z} \subset A_i$, which is an interval ipomset, then $w < y$ or $x < z$.    
	\end{itemize}
	Which concludes the proof.
\end{proof}

\section{Complementary proofs of Section~\ref{sec:relation-of-w-languages}}
\label{app:proof-relation-of-w-language}

\begin{proof}[Proof of Lemma~\ref{lem:ST-coh}]
    Let $A = (Q,E,I,F,\lambda)$ be an ST-automaton. 
    As $\lambda$ is coherent with $E$, all interfaces coincide in a word produced by $A$. Then it remains to show that the languages does not contain not allowed words.
    
    If $A$ accepts finite words, then even for a path of one states $q$, the label contains at least the letter $\id_{\lambda(q)}$. Thus $\varepsilon \not\in L(A)$, and $L(A) \subset \coh(\Xib)$.

    If $A$ accepts infinite words, then for every infinite path, one letter out of two is in $\Xib \setminus \Id$, so there is an infinite number of non-identity letters, thus $L(A) \subset \wcoh(\Xib)$.
\end{proof}

\begin{proof}[Proof of Lemma \ref{lem:equiv-evalutation}]
	Let $\alpha = (q_0,\phi_0,q_1,...)$ be an $\omega$-track in $X$ and $ST(\alpha) = (q_0,e_0,q_1,...)$ with $e_i = (q_i,P_i,q_{i\+1})$. By definition:
	\begin{itemize}
		\item
		If $(q_i,\phi_i,q_{i\+1}) = (\delta_{A_i}^0(q_{i\+1}),\arrO{A_i},q_{i\+1})$, then $P_i = \ev(q_i,\phi_i,q_{i\+1})$
		\item
		If $(q_i,\phi_i,q_{i\+1}) = (q_i,\arrI{A_i},\delta_{A_i}^1(q_i))$, then $P_i = \ev(q_i,\phi_i,q_{i\+1})$
	\end{itemize}
	We recall that $\lambda$ is coherent, so for $(q_i,e_i,q_{i\+1}) \in E$ we have $\lambda(q_i) = T_{P_i}$, $\lambda(q_{i\+1}) = S_{P_{i\+1}}$, hence:
	\begin{align*}
		\ev(\alpha)
		&=\ \ev(q_0,\phi_0,q_1) * \ev(q_1,\phi_1,q_2) * \cdots
		=\ P_0 * P_1 * \cdots 
		=\ \id_{\lambda(q_0)} * P_0 * \id_{\lambda(q_1)} * P_1 * \cdots \\
		&=\ \Psi( \id_{\lambda(q_0)} \cdot P_0 \cdot \id_{\lambda(q_1)} \cdot P_1 \cdot \ldots) 
		=\ \Psi( l(ST(\alpha)) ) \qedhere
	\end{align*}
\end{proof}

\begin{proof}[Proof of Lemma~\ref{lem:positive-expression}]
    For an expression $e_A$ of finite words, we replace any occurrence of the Kleene star with the formula $L^* = \varepsilon + L^+$, and ``developp'' the expression. For ST-automaton $A$, $ \varepsilon \not\in L(A)$, the resulting expression after developing is such that there is no $\varepsilon$ anymore.
    
    For an $\omega$-rational expression $e_{A'}$, we know that $e_{A'}$ can be a finite union of the form $\bigcup X.Y^\omega$ with $X$ and $Y$ rational sets of finite words (see Th.~3.2 of~\cite{PePin04}). It suffices to show that we can remove the occurrences of $^*$ and $\varepsilon$ for an expression $X.Y^\omega$. We replace any occurrences of the Kleene star by $L^* = \varepsilon + L^+$. Then we develop the expression such that:
    \begin{itemize}
        \item
            If $\varepsilon \in X$, $X = \varepsilon + X'$ (with $\varepsilon \not\in X'$) and $X.Y^\omega = Y^\omega + X'.Y^\omega = Y.Y^\omega + X'.Y^\omega$
        \item
            If $\varepsilon \in Y$, $Y = \varepsilon + Y'$ (with $\varepsilon \not\in Y'$) and because of the definition of $^\omega$, we have $X.Y^\omega = X.Y'\ ^\omega$
    \end{itemize}
    Thus we can transform $e_A$ into a positive expression $\overline{e}_A$ (resp. $e_{A'}$ into $\overline{e}_{A'}$). For a formula over infinite words, it is moreover still of the form $\bigcup X.Y^\omega$.
\end{proof}

\begin{proof}[Proof of Lemma~\ref{lem:ST-rat}]
	
	Let $A$ be an ST-automaton of finite words. The ST-automaton is finite, so by the classical Kleene theorem, $L(A)$ can be represented by a rational expression $e_A(U_1,...,U_n)$, written $e_A$ for simplification. We have $\mc{L}(e_A) = L(A)$. By Lem~\ref{lem:positive-expression} , we can suppose $e_A$ to be positive, and by Lem~\ref{lem:psi'-psi}, $\Psi(L(A)) =\Psi(\mc{L}(e_A)) = \Psi'(e_A)$. By Lem~\ref{lem:psi'-in-wrat}, $\Psi'(e_A) \in \rat$, so $\Psi(L(A)) \in \rat$.
	
	The same can be done for an ST-automaton of infinite words $A'$. By the $\omega$-Kleene theorem (for example shown in~\cite{PePin04}), $L(A')$ can be represented by a $\omega$-rational positive expression $e_{A'}$, with $\mc{L}(e_{A'}) = L(A')$, and we have $\Psi(L(A')) = \Psi(\mc{L}(e_{A'}) = \Psi'(e_{A'})$, and by Lem~\ref{lem:psi'-in-wrat}, $\Psi(L(A')) \in \wrat$.
\end{proof}

The following results are needed for the induction part of the proof of Prop.~\ref{prop:finite-rationality}:
\begin{lemma}
\label{lem:rat-eq-*}
    For $L,M \subset \iipoms$, $(L * M)\down = L\down *\sd M\down$.
\end{lemma}

\begin{proof}[Proof of Lemma~\ref{lem:rat-eq-*}]
    Let $P \in (L * M)\down$, there is $A \in L$ and $B \in M$ such that $P \subsu A * B$. As $L \subset L\down$ and $ M \subset M\down$, $A*B \in L\down * M\down$ so $P \in (L\down * M\down)\down = L\down *\sd M\down$ and $(L * M)\down \subset L\down *\sd M\down$. 

    Let $Q \in  L\down *\sd M\down$, there is $A \in L\down$ and $B \in M\down$ such that $Q \subsu A*B$. There is $A' \in L$ such that $A \subsu A'$ and $B' \in M$ such that $ B \subsu B'$. By Lem.~48 of~\cite{DBLP:journals/mscs/FahrenbergJSZ21}, $A*B \subsu A'*B'$ and by transitivity, $Q \subsu A'*B'$. But we have $A'*B' \in L*M$, so $Q \in (L * M)\down$ and $L\down *\sd M\down \subset (L * M)\down$.
\end{proof}

\begin{corollary}
\label{cor:rat-eq-+}
    For $L,M \subset \iipoms$, $(L^+)\down = (L\down)^{+\sd}$.
\end{corollary}

\begin{proof}[Proof of Proposition~\ref{prop:finite-rationality}]
We start by showing that $\rat\down \subset \Drat$ by induction:
\begin{itemize}
    \item
        $\emptyset\down = \emptyset \in \Drat$
    \item
        For $U \in \Xi$, $\set{U}\down = \set{U}$ (as $<_U$ is empty) and $\set{U} \in \Drat$ (by a finite use of $\parallel\sd$ and $\set{a \ibu}$ and $\set{\ibu a \ibu}$,  or $\set{\ibu a}$ and $\set{\ibu a \ibu}$, for $a \in \Sigma$)
    \item
        For $L,M \in \rat$, if $L\down,M\down \in \Drat$, then $(L+M)\down = L\down + M\down \in \Drat$
    \item
         For $L,M \in \rat$, if $L\down,M\down \in \Drat$, then $(L*M)\down = L\down *\sd M\down \in \Drat$ (by Lem.~\ref{lem:rat-eq-*})
    \item
         For $L \in \rat$, if $L\down \in \Drat$, then $(L^+)\down = (L\down)^{+\sd} \in \Drat$ (by Cor.~\ref{cor:rat-eq-+})
\end{itemize}
Then by induction $\rat\down \subset \Drat$.

Next we show that $\Drat \subset \rat$, which is just a more detailed version of the Remark 7.3 of~\cite{DBLP:conf/concur/FahrenbergJSZ22}. Let $L \in \Drat$, by Th.~\ref{th:kleene} there is a HDA $X$ such that $L = L(X)$. Let $ST(X)$ be its ST-automata, by Prop.~\ref{prop:eq-ST-HDA} $L(X) = \Psi(L(ST(X)))$, and by Lem.~\ref{lem:ST-rat}, $\Psi(L(ST(X))) \in \rat$. Thus $\Drat \subset \rat$.

But $L \in \Drat$, $L$ is down-closed (by Prop.~\ref{prop:HDA-down-closed} and Th.~\ref{th:kleene}) so as $L \in \rat$, $L = L\down \in \rat\down$. Then $\Drat \subset \rat\down$.
\end{proof}

\begin{proof}[Proof of Theorem~\ref{prop:buchi-wform}]
    We use the key property that language of a classical Büchi automaton $A = (Q,E,I,F)$ is of the form $\bigcup_{q_i \in I,\, q_f \in F} W(q_i,q_f).W(q_f,q_f)^\omega$, where $W(p,q)$ is the language of the Büchi automaton $(Q,E,\set{p},\set{q})$ (see~\cite{Farwer01} for more details).

    Let $(X,\bot_X,\top_X)$ be a Büchi $\omega$-HDA, and $ST(X) = (Q,E,I,F,\lambda)$ be its ST-automaton. Thus $L(ST(X)) = \mc{L}(\bigcup_{q_i \in I,\, q_f \in F} W(q_i,q_f).W(q_f,q_f)^\omega)$. To simplify, we write $W_1 = W(q_i,q_f)$ and $W_2 = W(q_f,q_f)$. As, $ W_2 = (W_2)^+ + \varepsilon$ so we have $L(ST(X)) = \mc{L} (\bigcup W_1.((W_2)^+)^\omega)$.

    But $\mc{L}(W_1)$ and $\mc{L}(W_2)$ are languages of ST-automata from HDAs (over finite tracks), so by Prop.~\ref{prop:finite-rationality} it exists $X,Y \subset \iipoms$ such that $\Psi(\mc{L}(W_1)) = X\down$ and $\Psi(\mc{L}(W_2)) = Y\down$. Then:
    \begin{align*}
        L(X) =&\ \Psi(L(ST(X))) \tag{Prop.~\ref{prop:eq-ST-wHDA}}
        =\ \Psi( \mc{L} (\bigcup W_1.((W_2)^+)^\omega) )\\
        =&\ \Psi'( \bigcup W_1.((W_2)^+)^\omega )  \tag{Lem.~\ref{lem:psi'-psi}}
        =\ \bigcup \Psi'(W_1) * (\Psi'(W_2)^+)^\omega \\
        =&\ \bigcup \Psi(\mc{L}(W_1)) * (\Psi(\mc{L}(W_2))^+)^\omega \tag{Lem.~\ref{lem:psi'-psi}}\\
        =&\ \bigcup X\down * ((Y\down)^+)^\omega
    \end{align*}
Which concludes the proof.
\end{proof}

\end{document}